\documentclass[submission,copyright,creativecommons]{eptcs}
\providecommand{\event}{NCL 2026} 

\usepackage{iftex}

\ifpdf
  \usepackage{underscore}         
  \usepackage[T1]{fontenc}        
\else
  \usepackage{breakurl}           
\fi

\newcommand{\D}{\mathfrak{D}}
\newcommand{\E}{\mathfrak{E}}
\usepackage{units}
\usepackage{amsthm,amsmath,amssymb,amsfonts}

  \newtheorem{definition}{\sc Definition}[section]
  \newtheorem{theorem}{\sc Theorem}[section]
  \newtheorem{lemma}{\sc Lemma}[section]
  \newtheorem{prop}{\sc Proposition}[section]
    \newtheorem{coro}{\sc Corollary}[section]
\newcommand{\bc}{\begin{center}}
	\newcommand{\ec}{\end{center}}
\usepackage{graphicx,hyperref}
\usepackage{bussproofs}
\usepackage{mathrsfs,multicol}
\usepackage{enumerate,longtable}
\usepackage{adjustbox}
\EnableBpAbbreviations
\newcommand{\A}{\mathscr{A}}
\newcommand{\ND}{\mathfrak{ND}}
\newcommand{\CNF}{{\bf CNF}}
\newcommand{\CNFi}{{\bf CNF}^\iota}
\newcommand{\INF}{{\bf INF}}
\newcommand{\INFi}{{\bf INF}^\iota}
\newcommand{\KNF}{{\bf NFK}_3}
\newcommand{\KNFi}{{\bf NFK}_3^\iota}
\newcommand{\KNFia}{{\bf NFK}_3^{\iota\ast}}

\newcommand{\Liota}{\mathcal{L}^\iota}
\newcommand{\Lbiota}{\mathcal{L}_\bot^\iota}

\newcommand{\io}{\iota xAx}
\newcommand{\Ex}[1]{\exists!#1}
\newcommand{\Ioa}{\io=a}
\newcommand{\NIoa}{\io\neq a}
\newcommand{\Bb}{Ab\leftrightarrow b=a}
\newcommand{\NBb}{\neg(Ab\leftrightarrow b=a)}
\newcommand{\ExNB}{\exists y\,\neg(Ay\leftrightarrow y=a)}

\title{Normalisation for a Neutral Free Logic with Definite Descriptions}
\author{Yaroslav Petrukhin
\institute{University of Lodz, Center for Philosophy of Nature,\\ Lindleya 3/5, 90-131, 
Lodz, Poland}
\email{yaroslav.petrukhin@gmail.com}
}
\def\titlerunning{Normalisation for a Neutral Free Logic with Definite Descriptions}
\def\authorrunning{Y. Petrukhin}
\begin{document}
\maketitle

\begin{abstract}
In this paper, we present a natural deduction system for a strong Kleene-based neutral free logic with definite descriptions. The system combines a three-valued propositional basis with free quantification, identity, existence, and rules for definite descriptions. The main result is a normalisation theorem for a proof-theoretically equivalent restricted formulation of the system; we also establish a corresponding modified subformula property for its fragment. 
\end{abstract}

\section{Introduction}
Definite descriptions are expressions of the form `the $F$' usually formalised by means of $\iota$-term-forming operator as $\iota x F$. If `the $F$' stands for a unique existing object, a definite description is called proper, otherwise it is improper. For example, 
``the author of \emph{Of Mice and Men}'' is a proper definite description,  while ``the author of \emph{Tale of the Troika}'' (uniqueness fails) and  ``Rachmaninoff's violin concerto'' (no such work exists) are improper definite descriptions.

Free logic is designed for languages in which singular terms need not denote. This makes it a natural setting for the proof theory of definite descriptions and was one of the motives of its development  by Hintikka \cite{Hintikka59} and Lambert \cite{Lambert2,Lambert3,Lambert4}. In negative free logic, atomic formulas containing non-denoting terms are false (in contrast to positive free logic, where they still might be true).  K\"urbis \cite{KurbisNormNeg,KurbisNormPos} recently proved normalisation theorems for natural deduction systems of intuitionistic and classical negative and positive free logics, both without and with definite descriptions.\footnote{See also K\"urbis' papers regarding a binary quantifier for definite descriptions in natural deduction systems \cite{KurbisBinary19,KurbisBinary20,KurbisBinary21}.} In negative free logics with descriptions, he uses Tennant's \cite{Tennant78} rules for $\iota$ and shows how to modify the natural deduction system so that the normalisation proof can proceed by induction on the degree of maximal formulas. The main technical difficulty is that substituting a complex $\iota$-term for a parameter may create new maximal formulas of higher, or uncontrolled, degree. Following an observation of Indrzejczak \cite{IndrzejczakCutFree}, this is handled by using a general introduction rule for identity and by restricting the relevant instantiating terms to atomic terms.

The present paper extends this line of work to a neutral free logic based on Strong Kleene logic \cite{Kleene38} $\bf K_3$. In neutral free logic,  atomic formulas containing non-denoting terms are neither true nor false. Early formal treatments of truth-value gaps in free logic are due to van Fraassen \cite{vanFraassen66,vanFraassen69} and Woodruff \cite{Woodruff70,Woodruff84}. These works investigate both direct
truth-value-gap semantics and supervaluational ways of recovering
classical logical behaviour in the presence of non-denoting terms.
Neutral free logic was subsequently developed in detail by Lehmann
\cite{Lehmann94,Lehmann02}. The present system instead takes the
three-valued Strong Kleene logic ${\bf K}_3$ as its propositional
basis. The system studied here combines a three-valued
propositional basis with free quantification, identity, the existence predicate $\exists!$, and rules for definite descriptions. Therefore, it requires not only Tennant's \cite{Tennant78}  positive rules for descriptive identity $\iota xF=t$, but also rules for
negative descriptive identity $\iota xF\neq t$ and for failure of
descriptive definedness $\neg\exists!(\iota xF)$. The formula $\iota xF\neq t$ says that the description is not identical with a given term $t$. By contrast, $\neg\exists!(\iota xF)$ says that the description itself fails to denote. Thus, improper descriptions are
treated in two connected but distinct ways: relative to a given object and  by the failure of descriptive definedness. In free logic, definite descriptions are usually characterized by Lambert's law: $\forall z\bigl(\iota x A x = z \leftrightarrow
\forall y(Ay \leftrightarrow y=z)\bigr)$. As a derivable principle, one could obtain $\exists!(\iota xAx)\leftrightarrow \exists z(\iota xAx=z)$. The rules for $\iota xF\neq t$ and $\neg\exists!(\iota xF)$ are motivated by the negative Lambert-style principle $\forall z\bigl(\iota xAx\neq z \leftrightarrow\exists y\neg(Ay\leftrightarrow y=z)\bigr)$ and by the negative definedness principle $\neg\exists!(\iota xAx)\leftrightarrow\neg\exists z(\iota xAx=z)$. Proposition~\ref{Lambert} shows that these principles, together with their positive counterparts, are derivable in negative and neutral free logics.

Definite descriptions in strong and weak Kleene-based neutral free logics have previously been investigated in \cite{IndrzejczakPetrukhin24Bisequent} by means of cut-free bisequent calculi and three-valued semantics\footnote{A sequent calculus treatment of definite descriptions in positive and
negative free logics is given in
\cite{IndrzejczakFreeDefDesc20,IndrzejczakModalDescriptions18,IndrzejczakNegativeFreeLogics25,KurbisProofSemantics21,KurbisBinary22}.}. The present paper takes a different proof-theoretic route. It develops a natural deduction system and studies normalisation rather than cut admissibility. Moreover, in \cite{IndrzejczakPetrukhin24Bisequent}, neutral free logics are based on Pavlovi\'c and Gratzl's \cite{PavlovicGratzl23} neutral free logics, which treat the existence predicate $\exists!$ as two-valued; in particular, the law of excluded middle is valid for it. In our logic, the law of excluded middle for $\exists!$ is not postulated as a separate rule. We expect the corresponding semantics for $\exists!$ to be three-valued. 

Related recent work investigates second-order definite descriptions in second-order logic \cite{PetrukhinJelia} and positive free logic \cite{Petrukhin2026SO} and definite descriptions in a paradefinite Russellian setting \cite{Petrukhin2026PRL}. The present paper differs from both in treating descriptions as singular terms of a first-order neutral free logic and in its proof-theoretic aim: a natural deduction normalisation theorem requiring new rules and reductions for negative descriptive identity and failure of descriptive definedness.

The main result is a normalisation theorem for a proof-theoretically
equivalent restricted formulation of the system; we also establish a
corresponding modified subformula property for its fragment. The proof combines K\"urbis' \cite{KurbisNormNeg} treatment of definite descriptions in negative free logic with Zimmermann's \cite{Zimmermann} method for handling Peirce's rule in
classical natural deduction as well as K\"urbis and Petrukhin's \cite{KurbisPetrukhin21} normalisation proof for some many-valued logics. We formulate the rules for negative descriptive identity and negative descriptive definedness, proving the necessary restriction lemmas for atomic instantiating terms, and extending the reduction procedures to the new maximal formulas and segments.

\section{Natural deduction systems}
The language $\Lbiota$ of $\INFi$ and $\CNFi$ consists of free variables called parameters denoted by $a,b,c,a_1,\ldots$, bound variables $x,y,z,x_1,\ldots$, constants $k,k_1,\ldots$, predicate letters $P,Q,R,S,P_1,\ldots$, identity $=$, the existence
predicate $\exists !$, the connectives $\bot, \wedge, \vee,\to,$ and quantifiers $ \exists, \forall$. The formula $\neg A$ is defined as
$A\to\bot$ and the formula $A\leftrightarrow B$ as $(A\to B)\wedge (B\to A)$. The description operator $\iota$ combines a variable $x$ and a formula $F$ to form the singular term $\iota xF$, in which $x$ is bound by $\iota x$. The terms of the language are constants, parameters, and $\iota$-terms. Constants and parameters are atomic terms, whereas $\iota$-terms are complex. $A,B,C,D,E,F,G,A_1,\ldots$ range over formulas and $t,u,t_1,\ldots$ over terms. We write $A^x_t$ for the result of substituting the term $t$ for all free occurrences of the variable $x$ in $A$. The distinction between parameters and bound variables guarantees that such substitutions are always well-defined. Prime formulas are those formed by applying predicate letters, $\exists!$, or $=$ to parameters and constants. Atomic formulas are those formed by applying these expressions to arbitrary terms, including complex ones. The language $\Liota$ of $\KNFi$ is obtained from $\Lbiota$ by dropping $\bot$ and taking $\neg$ as primitive. We write $t\not =u$ for $\neg t=u$. A formula $\neg A$, where $A$ is prime, is called a negative
prime formula. A formula $\neg A$, where $A$ is atomic, is called a
negative atomic formula.

Deductions are defined in Gentzen-Prawitz-style as certain kinds of trees labelled with formulas. As K\"{u}rbis does \cite{KurbisNormNeg}, we follow Troelstra and Schwichtenberg's conventions for the discharge, or closing, of assumptions \cite[Sec.~2.1.1]{TroelstraSchwichtenberg}. Assumptions are grouped into assumption classes, and a rule closes all assumptions in a given class. Such assumptions are enclosed in square brackets and marked with a numerical label; the same label is written next to the inference at which the class is discharged. Empty assumption classes, allowing vacuous discharge, are permitted. $\D,\mathfrak{E},\D',\D_1,\ldots$ range over deductions. We write $\Gamma\vdash_{\ND} A$, if there is a deduction of $A$ in a natural deduction system $\ND$ from some
of the formulas in $\Gamma$ as open assumptions.

Tennant's \cite[Sec.~7.10]{Tennant78} natural deduction system for intuitionistic negative free logic, $\INF$, consists of all the rules displayed in Figure~\ref{fig:ND1}, with the exception of $(\bot E_C)$ and the rules for $\iota$. The system $\INFi$ \cite[Sec.~7.10]{Tennant78} is obtained from $\INF$ by adding the rules for $\iota$. The system $\CNF$ \cite{KurbisNormNeg} for classical negative free logic is obtained from $\INF$ by adding $(\bot E_C)$; here $(\bot E)$ is treated as the special vacuous case of $(\bot E_C)$. Finally, $\CNFi$ \cite{KurbisNormNeg} is obtained from $\CNF$ by adding the rules for $\iota$.

The natural deduction system $\KNF$ for neutral free logic consists of all the rules displayed in Figure~\ref{fig:ND1}, except $(\bot E)$, $(\bot E_C)$, and the rules for $\iota$, together with all the rules displayed in Figure~\ref{fig:ND2}, except the rules for $\iota$. The system $\KNFi$ is obtained from $\KNF$ by adding the rules for $\iota$. Equivalently, $\KNFi$ contains all the rules displayed in Figures~\ref{fig:ND1} and~\ref{fig:ND2}, with the exception of $(\bot E)$ and $(\bot E_C)$. 

As pointed out by K\"{u}rbis \cite[p. 243]{KurbisNormNeg}, ``(\textit{AD}), the rule of atomic denotation, is typical for negative free
logic: an atomic sentence can only be assertible or true if all terms that occur in it refer.'' In $\KNFi$, we have its negative counterpart, $(\mathit{AD}')$, used, e.g., in a four-valued free logic by Antunes and Rodrigues \cite{AntunesRodrigues25}. Following K\"urbis \cite[p.~244]{KurbisNormNeg}, we treat $(\iota E_3)$ separately from $(AD)$ as the denotation-extraction rule for positive descriptive identity. On the negative side, denotation is extracted from ordinary negative atomic formulas by $(AD')$, from inequalities between atomic terms by $(\neq D)$, and from negative descriptive identities by $(\neg\iota E_1)$. Following K\"urbis \cite{KurbisNormNeg}, we denote the introduction rule for identity by $(=I^n)$, where the superscript $n$ stands for negative free logic.

As follows from \cite{Priest02,KurbisPetrukhin21}, the $\{\neg,\wedge,\vee\}$-fragment of $\KNFi$ is the natural deduction system for the three-valued Strong Kleene logic $\mathbf{K}_3$ \cite{Kleene38}; the $\{\neg,\wedge,\vee,\to\}$-fragment of $\KNFi$ corresponds to an implicative extension of $\mathbf{K}_3$ studied in \cite{SlupeckiPComp,avronnaturallogics} and known as \textbf{PComp} due to \cite{Popov09}; the corresponding $\{\neg,\wedge,\vee\}$- and $\{\neg,\wedge,\vee,\to\}$-fragments without the rule of ex falso quodlibet (\textit{EFQ}) yield Anderson--Belnap--Dunn's logic $\mathbf{FDE}$ \cite{AndersonBelnap63,Belnap77Useful,Belnap77Computer,Dunn76}, First Degree Entailment, and its implicative extension known as $\mathbf{Par}$ \cite{Popov89} and $\mathbf{HBe}$ \cite{avronnaturallogics}. As follows from \cite{Prawitz,KamideWansing15}, the $\{\neg,\wedge,\vee,\to\}$-fragment of $\KNFi$ without (\textit{EFQ}) and Peirce’s rule ($P$) yields Almukdad and Nelson's \cite{AlmukdadNelson84} \textbf{N4}. The fragment of $\KNF$ without (\textit{EFQ}) and rules for implication is a fragment of Antunes and Rodrigues' \textbf{FFDE} \cite{AntunesRodrigues25} (their logic also has additional rules for $\exists !$ to make its behavior classical). This fragment with the rules for implication, except ($P$), is a fragment of \textbf{FN4} \cite{AntunesRodrigues25}. The fragment of $\KNF$ without (\textit{EFQ}), ($P$), and ($\not= I$) is $\bf N4^{NF}$ \cite{PetrukhinNelsonianDescriptions}.


\begin{lemma}\cite[Lemma 6]{KurbisNormNeg}\label{sym}
The following rule (\textit{sym}) $\frac{t_1=t_2}{t_2=t_1}$ is derivable.
\end{lemma}

\begin{lemma}\label{notsym}
If $t_1$ and $t_2$ are atomic terms, the following rule ($\not=$\textit{sym}) $\frac{t_1\not =t_2}{t_2\not =t_1}$ is derivable.    
\end{lemma}
\begin{proof}
Use the rules $(\not =D)$, $(=I^n)$, and $(\neq I)$.
%
%
\end{proof}
 
\begin{prop}\label{Lambert}
    $\CNFi$ and $\KNFi$ prove the following formulas:
\begin{enumerate}[$(1)$]\itemsep=0pt
\begin{multicols}{2}
 \item $\forall z\bigl(\iota x A x = z \leftrightarrow
\forall y(Ay \leftrightarrow y=z)\bigr)$
    \item $\forall z\bigl(\iota x A x \not= z \leftrightarrow
\exists y \neg(Ay \leftrightarrow y=z)\bigr)$ 
\item $\exists!(\iota xAx)\leftrightarrow \exists z(\iota xAx=z)$
\item $
\neg\exists!(\iota xAx)\leftrightarrow \neg\exists z(\iota xAx=z)
$
\end{multicols}
\end{enumerate}    
\end{prop}
\begin{proof}
As an example, we prove (2), (3), and (4). Recall that in $\CNFi$, $\neg A$ is $A\to \bot$. 

(2) Let us prove $\forall z\bigl(\iota x A x \not= z \leftrightarrow
\exists y \neg(Ay \leftrightarrow y = z)\bigr)$ in $\CNFi$. 
Let $\D_4$ be the following deduction:
{\small\begin{prooftree}
    \AxiomC{$[\NBb]^4$}
\AxiomC{$[b=a]^3$}
\RL{$(AD)$}
\UnaryInfC{$\Ex b$}
\RL{$(\exists I)$}
\BinaryInfC{$\ExNB$}

\AxiomC{$[\neg\ExNB]^2$}
\LL{$({\to} E)$}
\BinaryInfC{$\bot$}
\RL{$(\bot E_C)^4$}
\UnaryInfC{$\Bb$}
\RL{$(\wedge E_2)$}
\UnaryInfC{$b=a\to Ab$}

\AxiomC{$[b=a]^3$}
\LL{$({\to} E)$}
\BinaryInfC{$Ab$}
\end{prooftree}}

Let $\D_5$ be the following deduction:
{\small\begin{prooftree}
\setlength{\topsep}{2pt}
\renewcommand{\proofSkipAmount}{\vskip0pt}
   \AxiomC{$[\NBb]^5$}
\AxiomC{$[\Ex b]^7$}
\RightLabel{$(\exists I)$}
\BinaryInfC{$\ExNB$}
\AxiomC{$[\neg\ExNB]^2$}
\LL{$(\to E)$}
\BinaryInfC{$\bot$}
\RightLabel{$(\bot E_C)^5$}
\UnaryInfC{$\Bb$}
\RightLabel{$(\wedge E_1)$}
\UnaryInfC{$Ab\to b=a$}
\AxiomC{$[Ab]^6$}
\LL{$({\to} E)$}
\BinaryInfC{$b=a$} 
\end{prooftree}}

Let $\D_6$ be the following deduction:
{\small\begin{prooftree}
\setlength{\topsep}{2pt}
\renewcommand{\proofSkipAmount}{\vskip0pt}
    \AxiomC{$[\Ioa]^9$}
\AxiomC{$[Ab]^{12}$}
\AxiomC{$[\Ex b]^{11}$}
\RightLabel{$(\iota E_2)$}
\TrinaryInfC{$b=a$}
\RightLabel{$({\to} I)^{12}$}
\UnaryInfC{$Ab\to b=a$}
\AxiomC{$[\Ioa]^9$}
\AxiomC{$[b=a]^{13}$}
\RightLabel{$(\iota E_1)$}
\BinaryInfC{$Ab$}
\RightLabel{$({\to} I)^{13}$}
\UnaryInfC{$b=a\to Ab$}
\RightLabel{$(\wedge I)$}
\BinaryInfC{$\Bb$}
\end{prooftree}}

Finally, we get the following deduction:
{\footnotesize \begin{prooftree}
\AxiomC{$[\Ex a]^0$}
\AXC{$\D_4$}
\AXC{$\D_5$}
\RightLabel{$(\iota I)^{3,6,7}$}
\TrinaryInfC{$\Ioa$}
\AxiomC{$[\NIoa]^1$}
\LL{$({\to} E)$}
\BinaryInfC{$\bot$}
\RightLabel{$(\bot E_C)^2$}
\UnaryInfC{$\ExNB$}
\RightLabel{$({\to} I)^1$}
\UnaryInfC{$\NIoa\to\ExNB$}
\AxiomC{$[\ExNB]^8$}
\AXC{$\D_6$}
\AxiomC{$[\NBb]^{10}$}
\RL{$({\to} E)$}
\BinaryInfC{$\bot$}
\RightLabel{$(\exists E)^{10,11}$}
\BinaryInfC{$\bot$}
\RightLabel{$(\bot E_C)^9$}
\UnaryInfC{$\NIoa$}
\RightLabel{$(\to I)^8$}
\UnaryInfC{$\ExNB\to\NIoa$}
\RightLabel{$(\wedge I)$, $\leftrightarrow\mathrm{def}$}
\BinaryInfC{$\NIoa\leftrightarrow\ExNB$}
\RightLabel{$(\forall I)^0$}
\UnaryInfC{$
\forall z\bigl(\iota xAx\neq z
\leftrightarrow
\exists y\,\neg(Ay\leftrightarrow y=z)\bigr)
$}
\end{prooftree}}

Let us prove
$\forall z\bigl(\iota x A x \not= z \leftrightarrow
\exists y \neg(Ay \leftrightarrow y = z)\bigr)$ in $\KNFi$.
Let $d$ be a fresh parameter. Let $\D_7$ stand for the following deduction:
{\footnotesize
\begin{prooftree}
\AxiomC{$[\iota xAx\neq d]^1$}
\AxiomC{$[Aa]^3$}
\AxiomC{$[d\neq a]^4$}
\RightLabel{$(\neq sym)$}
\UnaryInfC{$a\neq d$}
\RightLabel{$(\neg{\to}I)$}
\BinaryInfC{$\neg(Aa\to a=d)$}
\RightLabel{$(\neg{\wedge}I_1)$}
\UnaryInfC{$\neg(Aa\leftrightarrow a=d)$}
\AxiomC{$[\exists!a]^2$}
\LL{$(\exists I)$}
\BinaryInfC{$\exists y\,\neg(Ay\leftrightarrow y=d)$}
\AxiomC{$[d=b]^7$}
\RightLabel{$(sym)$}
\UnaryInfC{$b=d$}
\AxiomC{$[\neg Ab]^6$}
\RightLabel{$(\neg{\to}I)$}
\BinaryInfC{$\neg(b=d\to Ab)$}
\RightLabel{$(\neg{\wedge}I_2)$}
\UnaryInfC{$\neg(Ab\leftrightarrow b=d)$}
\AxiomC{$[\exists!b]^5$}
\LL{$(\exists I)$}
\BinaryInfC{$\exists y\,\neg(Ay\leftrightarrow y=d)$}
\RightLabel{$(\neg\iota E_2)^{2,\ldots,7}$}
\TrinaryInfC{$\exists y\,\neg(Ay\leftrightarrow y=d)$}
\RightLabel{$(\to I)^1$}
\UnaryInfC{$\iota xAx\neq d\to\exists y\,\neg(Ay\leftrightarrow y=d)$}
\end{prooftree}
}

Let $\D_8$ be the following deduction:
{\small\begin{prooftree}
\setlength{\topsep}{2pt}
\renewcommand{\proofSkipAmount}{\vskip0pt}
\AxiomC{$[\exists!d]^0$}
\AxiomC{$[\exists!c]^9$}
\AxiomC{$[\neg(Ac\to c=d)]^{11}$}
\RightLabel{$(\neg{\to}E_1)$}
\UnaryInfC{$Ac$}
\AxiomC{$[\neg(Ac\to c=d)]^{11}$}
\RightLabel{$(\neg{\to}E_2)$}
\UnaryInfC{$c\neq d$}
\RightLabel{$(\neq sym)$}
\UnaryInfC{$d\neq c$}
\RightLabel{$(\neg\iota I_1)$}
\QuaternaryInfC{$\iota xAx\neq d$}
\end{prooftree}}

Let $\D_9$ be the following deduction:
{\small\begin{prooftree}
\setlength{\topsep}{2pt}
\renewcommand{\proofSkipAmount}{\vskip0pt}
\AxiomC{$[\exists!d]^0$}
\AxiomC{$[\exists!c]^9$}
\AxiomC{$[\neg(c=d\to Ac)]^{12}$}
\RightLabel{$(\neg{\to}E_2)$}
\UnaryInfC{$\neg Ac$}
\AxiomC{$[\neg(c=d\to Ac)]^{12}$}
\RightLabel{$(\neg{\to}E_1)$}
\UnaryInfC{$c=d$}
\RightLabel{$(sym)$}
\UnaryInfC{$d=c$}
\RightLabel{$(\neg\iota I_2)$}
\QuaternaryInfC{$\iota xAx\neq d$}
\end{prooftree}}

Finally, we obtain the following deduction:
{\small \begin{prooftree}
\setlength{\topsep}{2pt}
\renewcommand{\proofSkipAmount}{\vskip0pt}
\AXC{$\D_7$}
\AxiomC{$[\exists y\,\neg(Ay\leftrightarrow y=d)]^8$}
\AxiomC{$[\neg(Ac\leftrightarrow c=d)]^{10}$}
\AXC{$\D_8$}
\AXC{$\D_9$}
\RightLabel{$(\neg{\wedge}E)^{11,12}$}
\TrinaryInfC{$\iota xAx\neq d$}
\RightLabel{$(\exists E)^{9,10}$}
\BinaryInfC{$\iota xAx\neq d$}
\RightLabel{$(\to I)^8$}
\UnaryInfC{$\exists y\,\neg(Ay\leftrightarrow y=d)\to\iota xAx\neq d$}
\RightLabel{$(\wedge I)$, $\leftrightarrow\mathrm{def}$}
\BIC{$\iota xAx\neq d\leftrightarrow\exists y\,\neg(Ay\leftrightarrow y=d)$}
\RightLabel{$(\forall I)^0$}
\UnaryInfC{$\forall z\bigl(\iota xAx\neq z\leftrightarrow\exists y\,\neg(Ay\leftrightarrow y=z)\bigr)$}
\end{prooftree}}

(3) Let us prove $
\exists!(\iota xAx)\leftrightarrow \exists z(\iota xAx=z)
$ in $\CNFi$ and $\KNFi$ (the proof is the same in both logics). 
{\small\begin{prooftree}
\setlength{\topsep}{2pt}
\renewcommand{\proofSkipAmount}{\vskip0pt}
\AxiomC{$[\Ex({\io})]^1$}
\RightLabel{$(=I^n)$}
\UnaryInfC{$\io=\io$}

\AxiomC{$[\Ex({\io})]^1$}
\RightLabel{$(\exists I)$}
\BinaryInfC{$\exists z(\io=z)$}

\RightLabel{$(\to I)^1$}
\UnaryInfC{$\Ex({\io})\to \exists z(\io=z)$}

\AxiomC{$[\exists z(\io=z)]^2$}

\AxiomC{$[\io=a]^4$}
\LL{$(sym)$}
\UnaryInfC{$a=\io$}

\AxiomC{$[\Ex a]^3$}

\LL{$(=E)$}
\BinaryInfC{$\Ex({\io})$}

\RightLabel{$(\exists E)^{3,4}$}
\BinaryInfC{$\Ex({\io})$}

\RightLabel{$(\to I)^2$}
\UnaryInfC{$\exists z(\io=z)\to \Ex({\io})$}

\RightLabel{$(\wedge I)$}
\BinaryInfC{$
(\Ex({\io})\to \exists z(\io=z))
\wedge
(\exists z(\io=z)\to \Ex({\io}))
$}

\RightLabel{$\leftrightarrow$ def}
\UnaryInfC{$
\exists!(\iota xAx)\leftrightarrow \exists z(\iota xAx=z)
$}
\end{prooftree}}

(4) Let us prove $\neg\exists!(\iota xAx)\leftrightarrow
\neg\exists z(\iota xAx=z)$ in $\CNFi$. The two implications are derived as follows:

\begin{center}
\begin{adjustbox}{max width=\textwidth}
\begin{tabular}{c@{\qquad}c}
\AxiomC{$[\exists z(\io=z)]^2$}
\AxiomC{$[\io=a]^4$}
\LL{$(sym)$}
\UnaryInfC{$a=\io$}
\AxiomC{$[\Ex a]^3$}
\LL{$(=E)$}
\BinaryInfC{$\Ex(\io)$}
\AxiomC{$[\neg\Ex(\io)]^1$}
\LL{$(\to E)$}
\BinaryInfC{$\bot$}
\RightLabel{$(\exists E)^{3,4}$}
\BinaryInfC{$\bot$}
\RightLabel{$(\to I)^2$}
\UnaryInfC{$\neg\exists z(\io=z)$}
\RightLabel{$(\to I)^1$}
\UnaryInfC{$\neg\Ex(\io)\to\neg\exists z(\io=z)$}
\DisplayProof

&

\AxiomC{$[\Ex(\io)]^6$}
\LL{$(=I^n)$}
\UnaryInfC{$\io=\io$}
\AxiomC{$[\Ex(\io)]^6$}
\LL{$(\exists I)$}
\BinaryInfC{$\exists z(\io=z)$}
\AxiomC{$[\neg\exists z(\io=z)]^5$}
\LL{$(\to E)$}
\BinaryInfC{$\bot$}
\LL{$(\to I)^6$}
\UnaryInfC{$\neg\Ex(\io)$}
\LL{$(\to I)^5$}
\UnaryInfC{$\neg\exists z(\io=z)\to\neg\Ex(\io)$}
\DisplayProof
\end{tabular}
\end{adjustbox}
\end{center}

Let us prove $
\neg\exists!(\iota xAx)\leftrightarrow \neg\exists z(\iota xAx=z)
$ in $\KNFi$. 

{\small\begin{prooftree}
\AxiomC{$[\neg\Ex({\io})]^1$}
\AxiomC{$[\Ex a]^2$}
\RightLabel{$(\neg\exists!\iota E)$}
\BinaryInfC{$\io\neq a$}
\RightLabel{$(\neg\exists I)^2$}
\UnaryInfC{$\neg\exists z(\io=z)$}
\RightLabel{$(\to I)^1$}
\UnaryInfC{$\neg\Ex({\io})\to \neg\exists z(\io=z)$}
\AxiomC{$[\neg\exists z(\io=z)]^3$}
\AxiomC{$[\Ex a]^4$}
\RightLabel{$(\neg\exists E)$}
\BinaryInfC{$\io\neq a$}
\RightLabel{$(\neg\exists!\iota I)^4$}
\UnaryInfC{$\neg\Ex({\io})$}
\RightLabel{$(\to I)^3$}
\UnaryInfC{$\neg\exists z(\io=z)\to \neg\Ex({\io})$}
\RightLabel{$(\wedge I)$, $\leftrightarrow$ def}
\BIC{$\neg\exists!(\iota xAx)\leftrightarrow \neg\exists z(\iota xAx=z)$}
\end{prooftree}}
\end{proof}
\begin{prop}\label{NegRules}
 All the rules from Figure \ref{fig:ND2}, except ($\not=D$), ($\mathit{AD}'$), ($=E'$), ($\neg\iota E_1$), and ($\neg\iota E_2$), are derivable in $\CNFi$.   
\end{prop}

\begin{figure}[t]
\begin{center}
\small
	($\wedge I$) $ \dfrac{A\quad B}{A \wedge B} $ \quad
	($\wedge E_{1} $) $ \dfrac{A \wedge B}{A} $ \quad
	($\wedge E_{2} $) $ \dfrac{A \wedge B}{B} $ \quad
	($\vee I_{1}$) $ \dfrac{A}{A \vee B} $ \quad
	($\vee I_{2}$) $ \dfrac{B}{A \vee B} $ \quad    
(${\to} E$) $ \dfrac{A\to B\quad A}{B} $ \quad
    ($\bot E$) $\dfrac{\bot}{A}$ \quad
\end{center}
\begin{center}\small
($\vee E^{i,j}$) $ \dfrac{\begin{matrix}
		& [A]^i & [B]^j \\
 & \D_1 & \D_2\\
		A \vee B &  C &  C 
		\end{matrix}}{C} $\quad 
($\to I^i$) $ \dfrac{\begin{matrix}
		 [A]^i  \\
         \D\\
		 B  
		\end{matrix}}{A\to B} $\quad    
($\forall I^i$) $ \dfrac{\begin{matrix}
		 [\exists!a]^i  \\
         \D\\
		 A^x_a  
		\end{matrix}}{\forall x A} $\quad 
($\exists E^{i,j}$) $ \dfrac{\begin{matrix}
		& [A^x_a]^i \; [\exists!a]^j \\
 & \D\\
\exists x A &  C  
		\end{matrix}}{C} $\quad 
        ($\bot E_C^i$) $ \dfrac{\begin{matrix}
		 [\neg A]^i  \\
         \D\\
		 \bot  
		\end{matrix}}{A} $\quad 
\end{center}
\begin{center}\small
($\iota I^{i,j,k}$) 
$ \dfrac{
\begin{matrix}
        & [a=t]^i & [F^x_a]^j\;[\exists!a]^k \\
        & \D_1     & \D_2 \\
\exists!t & F^x_a & a=t
\end{matrix}
}{
\iota xF=t
} $ \quad
($\iota E_{1}$) 
$ \dfrac{\iota xF=t\quad u=t}{F^x_u} $
\quad
($\iota E_{2}$) 
$ \dfrac{\iota xF=t\quad F^x_u\quad \exists!u}{u=t} $
\end{center}
\begin{center}\small
($\iota E_{3}$) 
$ \dfrac{\iota xF=t}{\exists!t} $\quad
($\forall E$) $ \dfrac{\forall x A\quad \exists!t}{A^x_t} $ \quad    ($\exists I$) $ \dfrac{A^x_t\quad \exists!t}{\exists x A} $ \quad  
($=I^n$) $ \dfrac{\exists!t}{t=t} $ \quad 
($= E$) $ \dfrac{t_1=t_2\quad \A^x_{t_1}}{\A^x_{t_2}} $ \quad  
(\textit{AD}) $ \dfrac{Rt_1\ldots t_n}{\exists!t_i} $ \quad 
\end{center}
where $a$ is a fresh parameter, $\A$ is an atomic formula, and $R$ is a predicate symbol, including $=$.  
    \caption{The rules of $\CNFi$ and $\INFi$. }
    \label{fig:ND1}
\end{figure}

\begin{figure}[t]
\begin{center}
{\small ($\mathit{EFQ}$) $\dfrac{\neg A\quad A}{B}$\;
($\neg\neg I$) $\dfrac{A}{\neg\neg A}$\;
($\neg\neg E$) $\dfrac{\neg\neg A}{A}$\;
($\neg{\to}I$) $\dfrac{A\quad \neg B}{\neg(A\to B)}$\;
($\neg{\to}E_1$) $\dfrac{\neg(A\to B)}{A}$\;
($\neg{\to}E_2$) $\dfrac{\neg(A\to B)}{\neg B}$\;}
\end{center}
\begin{center}\small
	(${\neg\vee} I$) $ \dfrac{\neg A\quad\neg B}{\neg(A \vee B)} $ \quad
	(${\neg\vee} E_{1} $) $ \dfrac{\neg(A \vee B)}{\neg A} $ \quad
	(${\neg\vee} E_{2} $) $ \dfrac{\neg(A \vee B)}{\neg B} $ \quad
	(${\neg\wedge} I_{1}$) $ \dfrac{\neg A}{\neg(A \wedge B)} $ \quad
	(${\neg\wedge} I_{2}$) $ \dfrac{\neg B}{\neg(A \wedge B)} $ \quad    
\end{center}
\begin{center}  \small
(${\neg\wedge} E^{i,j}$) $ \dfrac{\begin{matrix}
		& [\neg A]^i & [\neg B]^j \\
 & \D_1 & \D_2\\
		\neg(A \wedge B) &  C &  C 
		\end{matrix}}{C} $\quad   
        ($P^i$) $ \dfrac{\begin{matrix}
		 [A\to B]^i  \\
         \D\\
		 A  
		\end{matrix}}{A} $\quad  
($\neg\exists I^i$) $ \dfrac{\begin{matrix}
		 [\exists!a]^i  \\
         \D\\
		 \neg A^x_a  
		\end{matrix}}{\neg\exists x A} $\quad 
($\neg\forall E^{i,j}$) $ \dfrac{\begin{matrix}
		& [\neg A^x_a]^i \; [\exists!a]^j \\
 & \D\\
\neg\forall x A &  C  
		\end{matrix}}{C} $\quad         
\end{center}
\begin{center} \small
($\neg\exists E$) $ \dfrac{\neg\exists x A\quad \exists!t}{\neg A^x_t} $ \quad 
($\neg\forall I$) $ \dfrac{\neg A^x_t\quad \exists!t}{\neg\forall x A} $ \quad     
($\neg\iota I_1$) $\dfrac{\exists!t\quad \exists!u\quad F^x_u\quad t\not=u}{\iota x F\not=t}$\quad 
($\neg\iota I_2$) $\dfrac{\exists!t\quad \exists!u\quad \neg F^x_u\quad t=u}{\iota x F\not=t}$\quad 
\end{center}
\begin{center}\small
($\not=D$) $\dfrac{t_1\not=t_2}{\exists!t_i}$\;
($\mathit{AD}'$) $ \dfrac{\neg Rt_1\ldots t_n}{\exists!t_i} $ \; 
($= E'$) $ \dfrac{t_1=t_2\quad \neg\A^x_{t_1}}{\neg\A^x_{t_2}} $ \;
($\not= I$) $ \dfrac{\A^x_{t_1}\quad \neg \A^x_{t_2}}{t_1\not=t_2} $ \;   
($\neg\exists!\iota E$) $ \dfrac{  \neg\exists!(\iota xF) \quad\exists!t}{\iota xF\neq t} $
\end{center}
\begin{center}\small
    ($\neg\iota E_{1}$) 
$ \dfrac{\iota xF\not=t}{\exists!t} $\;
  ($\neg\iota E_{2}^{i_1,\ldots,i_6}$)   $ \dfrac{
\begin{matrix}
& [\exists!a]^{i_1}[F^x_a]^{i_2}[t\not=a]^{i_3} & [\exists!b]^{i_4}[\neg F^x_b]^{i_5}[t=b]^{i_6} \\
       & \D_1     & \D_2 \\
\iota xF\not=t &  C & C
\end{matrix}
}{C} $ \;
($\neg\exists!\iota I^i$)
$ \dfrac{\begin{matrix}
        [\exists!a]^i \\
        \D \\
        \iota xF\neq a
\end{matrix}}{
        \neg\exists!(\iota xF)
} $
\end{center}
where $a$ and $b$ are fresh parameters, $R$ is an ordinary predicate
letter, $\A$ is an atomic formula. In $(\neq I)$, the principal
symbol of $\A$ is either an ordinary predicate letter or $=$; if it is $=$, then $t_2$ is an atomic term. In
$(\neq D)$, $t_i$ is an atomic term.
    \caption{The specific rules of $\KNFi$.}
    \label{fig:ND2}
\end{figure}

\section{Normalisation}\label{Norm}
For the normalisation proof, we combine Zimmermann's \cite{Zimmermann} strategy for proving normalisation with K\"urbis' \cite{KurbisNormNeg} normalisation proof for $\INFi$. Zimmermann's strategy was developed for classical first-order logic and relies essentially on Peirce's rule $(P)$. It is designed to overcome the well-known difficulties concerning disjunction and the existential quantifier identified by Prawitz \cite{Prawitz}. This strategy has already been successfully applied to many-valued logics in \cite{Petrukhin19,KurbisPetrukhin21}.

\begin{definition}
The degree $d(A)$ of a formula $A$ and the degree $d(t)$ of a term $t$
are defined by simultaneous induction as follows.
\begin{enumerate}\itemsep=0pt
\item If $t$ is an atomic term, then $d(t)=0$. If $t$ is a complex term $\iota xF$, then $d(\iota xF)=d(F)+1$.

\item If $A$ is a prime formula $Rt_1\ldots t_n$, where $R$ is an ordinary predicate letter, then $d(A)=0$. If $A$ is a prime formula of the form $\exists!t$ or $t_1=t_2$, then $d(A)=1$.

\item If $A$ is the negation of a prime formula, then $d(A)$ is defined as in the previous clause. Thus, 
$d(\neg Rt_1\ldots t_n)=0$ for an ordinary predicate letter $R$, while $d(\neg\exists!t)=1$ and $d(t_1\neq t_2)=1$.

\item If $A$ is an atomic formula $Rt_1\ldots t_n$, where $R$ is an ordinary predicate letter, then
$d(A)=d(t_1)+\ldots+d(t_n)$. If $A$ is an atomic formula of the form $\exists!t$ or $t_1=t_2$, then $d(\exists!t)=d(t)+1$ and $d(t_1=t_2)=d(t_1)+d(t_2)+1$.

\item If $A$ is the negation of an atomic formula, then $d(A)$ is defined as in the previous clause, except for formulas of the form
$\neg\exists!(\iota xF)$. Thus, $d(\neg Rt_1\ldots t_n)=d(t_1)+\ldots+d(t_n)$ for an ordinary predicate letter $R$, $d(\neg\exists!t)=d(t)+1$, and $d(t_1\neq t_2)=d(t_1)+d(t_2)+1$.

\item Formulas of the form $\neg\exists!(\iota xF)$ are assigned one
additional degree:
$d(\neg\exists!(\iota xF))=d(\iota xF)+2$.

\item If $A$ is $\neg B$ and $B$ is not atomic, then $d(A)=d(B)+1$.

\item If $A$ is $B\wedge C$, $B\vee C$, or $B\to C$, then
$d(A)=d(B)+d(C)+1$.

\item If $A$ is $\forall xB$ or $\exists xB$, then $d(A)=d(B)+2$.
\end{enumerate}
\end{definition}

Notice that Clause~7 applies to iterated negations: for example,
$\neg\neg Rt$ is not the negation of an atomic formula, even when $\neg Rt$ is one. The special clause for $\neg\exists!(\iota xF)$ is needed for the reduction involving $(\neg\exists!\iota I)$ followed by $(\neg\exists!\iota E)$. Since $t$ in $(\neg\exists!\iota E)$ is restricted to atomic terms, we have $d(\iota xF\neq t)=d(\iota xF)+1$, whereas $d(\neg\exists!(\iota xF))=d(\iota xF)+2$. Hence, the reduction replaces a maximal formula by one of strictly lower degree. 
\begin{definition}\cite[Definition 3]{KurbisNormNeg}
The \emph{major premise of an elimination rule} is the formula that displays the
connective or $\iota$ in the general statement of the rule, here always the leftmost premise. All
others are \emph{minor premises}.    
\end{definition}
\begin{definition}
The rules $(\vee E), (\exists E), (\neg\wedge E), (\neg\forall E), (\neg\iota E_2)$ are called \emph{general elimination rules}, or \emph{gel-rules}. Their conclusions are inherited from minor deductions.    
\end{definition}
 Vacuous discharge applications of gel-rules are eliminated whenever they are created by a reduction. If a minor deduction of a gel-rule does not depend on the assumptions discharged by that rule, the application is replaced by that minor deduction. These empty conversions are not listed separately below.
 \begin{lemma}\cite[Lemma 4]{KurbisNormNeg}\label{InGder}
The rule $(=I^{nG})$ displayed below is derivable given $(=I^n)$, $(\exists I)$, and $(\exists E)$, where the parameter $a$ does not occur in $t$, $C$, nor any open assumptions of $\D$ other than those of the assumption class of $a=t$:
\begin{center}
\AxiomC{$\exists!t$}
\AxiomC{$[a=t]^i$}
\noLine
\UnaryInfC{$\D$}
\noLine
\UnaryInfC{$C$}
\RightLabel{$(=I^{nG})^i$}
\BinaryInfC{$C$}
\DisplayProof
\end{center}
\end{lemma}
\begin{definition}
 By \emph{introduction rules} we mean the rules explicitly introducing the principal connective or operator; in particular, $(AD)$, $(AD')$, $(\neq D)$, $(EFQ)$, and $(P)$ are not counted as introduction rules. The rule $(=I^{nG})$ is called a \emph{general} introduction rule.
\end{definition}
 
\begin{definition}
A \emph{segment} is a sequence of formula occurrences
$A_1,\ldots,A_n$ such that $A_1$ is not the conclusion of a gel-rule nor of $(=I^{nG})$, $A_n$ is not a minor premise of a gel-rule nor the premise $C$ of an application of $(=I^{nG})$, and, for every $i<n$, $A_i$ is either a minor premise of a gel-rule or the premise $C$ of an application of $(=I^{nG})$, and $A_{i+1}$ is the conclusion of that application.  The \emph{length} of a segment is the number of formula occurrences in it. The \emph{degree} of a segment is the degree of the formula occurring on it. A segment is \emph{maximal} if its last formula is the major premise of an elimination rule and, if its length is $1$, its sole formula is the conclusion of an introduction rule.  A \emph{maximal formula} is a formula occurrence that is the conclusion of an introduction rule and the major premise of an elimination rule.
\end{definition}

We assume the usual pure variable condition for rules with restrictions on parameters: the parameter introduced by such a rule occurs only for the purposes of that application of the rule. 
The condition ensures that no variable clashes arise in the transformations used in the normalisation proof. We write $\D^u_t$ for the result of replacing every occurrence of the term
$u$ in the deduction $\D$ by the term $t$.

As noticed in \cite{KurbisNormNeg},  $\iota$-terms create a difficulty for the normalisation. Since the degree of a formula must take into account not only logical constants and quantifiers, but also the complexity of terms, replacing a parameter by a complex $\iota$-term may increase the degree of formulas occurring in a deduction. As a result, the standard reductions for maximal
formulas of the forms $\forall xA$ and $\exists xA$ (in the case of $\KNFi$, of the forms $\neg\forall xA$ and $\neg\exists xA$ as well) may introduce new maximal formulas of higher degree than the ones they eliminate.  The difficulty can be overcome by using an observation due to Indrzejczak \cite{IndrzejczakCutFree}:
the rules $(\forall E)$ and $(\exists I)$ (in the case of $\KNFi$, $(\neg\exists E)$ and $(\neg\forall I)$ as well) may be restricted to atomic instantiating terms, using the  general introduction rule (in the sense of Negri and von Plato \cite{NegriPlato01} and Milne \cite{Milne10}) for identity $(=I^{nG})$. 

\begin{lemma}\label{lem:efq}
(\textit{EFQ}) may be restricted to prime and negative prime conclusions.
\end{lemma}
\begin{proof}
    Similarly to \cite[Lemma 2]{KurbisNormNeg}, where ($\bot E$) is proven to be restricted to prime conclusions. 
\end{proof}

\begin{lemma}\cite[Lemma 5]{KurbisNormNeg}\label{Ider}
    Given $(= E)$ and $(= I^{nG} )$, $(= I^n)$ is derivable.
\end{lemma}

\begin{lemma}\cite[Lemma 7]{KurbisNormNeg}\label{restr:qua}
Given $(= I^{nG} )$, $(\forall E)$ and $(\exists I)$ may be restricted to atomic terms.     
\end{lemma}
\begin{lemma}\label{restr:negqua}
Given $(= I^{nG} )$, $(\neg\forall I)$ and $(\neg\exists E)$ may be restricted to atomic terms. 
\end{lemma}
\begin{proof}
 Any application of $(\neg\exists E)$ and $(\neg\forall I)$
where $t$ is complex can be replaced by the following: 
\begin{center}
\begin{adjustbox}{max width=\textwidth}
\begin{tabular}{c@{\qquad}c}
\AxiomC{$\D_1$}
\noLine
\UnaryInfC{$\exists!t$}
\AxiomC{$\D_2$}
\noLine
\UnaryInfC{$\neg\exists xA$}
\AxiomC{$[a=t]^i$}
\RightLabel{$(AD)$}
\UnaryInfC{$\exists!a$}
\RightLabel{$(\neg\exists E)$}
\BinaryInfC{$\neg A^x_a$}
\AxiomC{$[a=t]^i$}
\LL{$(=E')$}
\BinaryInfC{$\neg A^x_t$}
\RightLabel{$(=I^{nG})^i$}
\BinaryInfC{$\neg A^x_t$}
\DisplayProof

&

\AxiomC{$\D_3$}
\noLine
\UnaryInfC{$\exists!t$}
\AxiomC{$[a=t]^i$}
\LL{$(sym)$}
\UnaryInfC{$t=a$}
\AxiomC{$\D_4$}
\noLine
\UnaryInfC{$\neg A^x_t$}
\LL{$(=E')$}
\BinaryInfC{$\neg A^x_a$}
\AxiomC{$[a=t]^i$}
\RightLabel{$(AD)$}
\UnaryInfC{$\exists!a$}
\RightLabel{$(\neg\forall I)$}
\BinaryInfC{$\neg\forall xA$}
\RightLabel{$(=I^{nG})^i$}
\BinaryInfC{$\neg\forall xA$}
\DisplayProof
\end{tabular}
\end{adjustbox}
\end{center}
where $a$ is a fresh parameter.
\end{proof}

The same problem arises for the rules for $\iota$. If the term substituted for a parameter is complex, the resulting deductions may contain maximal formulas of higher degree than the formula removed by the corresponding reduction. The problem is solved by extending Indrzejczak's \cite{IndrzejczakCutFree} observation to the rules for $\iota$: given $(=I^{nG})$, the terms occurring in $\iota$ rules can be restricted to atomic~terms.

\begin{lemma}\cite[Lemma 8]{KurbisNormNeg}\label{restr:iota}
Given $(=I^{nG})$ and $(=E)$, $t$ and $u$ in
$(\iota I)$, $(\iota E_1)$ and $(\iota E_2)$ can be restricted to
atomic terms.
\end{lemma}
\begin{lemma}\label{restr:negiota}
 Given $(=I^{nG})$, $(=E)$, $(=E')$, and $(\neq I)$,
the terms $t$ and $u$ in $(\neg\iota I_1)$, $(\neg\iota I_2)$,
$(\neg\iota E_2)$, and $(\neg\exists!\iota E)$ can be restricted to
atomic terms.
\end{lemma}
\begin{proof}
Replace an application of $(\neg\exists!\iota E)$ in which $t$ is complex as follows, where $a$ is a fresh parameter:
\begin{center}
\AxiomC{$\D$}
\noLine
\UnaryInfC{$\exists!t$}
\AxiomC{$[a=t]^i$}
\AxiomC{$\E$}
\noLine
\UnaryInfC{$\neg\exists!(\iota xF)$}
\AxiomC{$[a=t]^i$}
\RightLabel{$(AD)$}
\UnaryInfC{$\exists!a$}
\RightLabel{$(\neg\exists!\iota E)$}
\BinaryInfC{$\iota xF\neq a$}
\RightLabel{$(=E')$}
\BinaryInfC{$\iota xF\neq t$}
\RightLabel{$(=I^{nG})^i$}
\BinaryInfC{$\iota xF\neq t$}
\DisplayProof
\end{center}
  
Replace an application of $(\neg\iota I_1)$ in which $t$ is
complex by the following, where $a$ is a fresh parameter:
\begin{center}
\AxiomC{$\D_t$}
\noLine
\UnaryInfC{$\exists!t$}
\AxiomC{$[a=t]^i$}
\AxiomC{$[a=t]^i$}
\RightLabel{$(AD)$}
\UnaryInfC{$\exists!a$}
\AxiomC{$\D_u$}
\noLine
\UnaryInfC{$\exists!u$}
\AxiomC{$\D$}
\noLine
\UnaryInfC{$F^x_u$}
\AxiomC{$\D'$}
\noLine
\UnaryInfC{$t\neq u$}
\AxiomC{$[a=t]^i$}
\RightLabel{$(sym)$}
\UnaryInfC{$t=a$}
\RightLabel{$(=E')$}
\BinaryInfC{$a\neq u$}
\RightLabel{$(\neg\iota I_1)$}
\QuaternaryInfC{$\iota xF\neq a$}
\RightLabel{$(=E')$}
\BinaryInfC{$\iota xF\neq t$}
\RightLabel{$(=I^{nG})^i$}
\BinaryInfC{$\iota xF\neq t$}
\DisplayProof
\end{center}

If $u$ is complex, replace it by a fresh parameter $b$ (if both $t$ and $u$ are complex, apply these two replacements successively):
\begin{center}
\begin{adjustbox}{max width=\textwidth}
\AxiomC{$\D_u$}
\noLine
\UnaryInfC{$\exists!u$}
\AxiomC{$\D_t$}
\noLine
\UnaryInfC{$\exists!t$}
\AxiomC{$[b=u]^i$}
\RightLabel{$(AD)$}
\UnaryInfC{$\exists!b$}
\AxiomC{$\D$}
\noLine
\UnaryInfC{$F^x_u$}
\AxiomC{$[b=u]^i$}
\RightLabel{$(sym)$}
\UnaryInfC{$u=b$}
\RightLabel{$(=E)$}
\BinaryInfC{$F^x_b$}
\AxiomC{$\D'$}
\noLine
\UnaryInfC{$t\neq u$}
\AxiomC{$[b=u]^i$}
\RightLabel{$(sym)$}
\UnaryInfC{$u=b$}
\RightLabel{$(=E')$}
\BinaryInfC{$t\neq b$}
\RightLabel{$(\neg\iota I_1)$}
\QuaternaryInfC{$\iota xF\neq t$}
\RightLabel{$(=I^{nG})^i$}
\BinaryInfC{$\iota xF\neq t$}
\DisplayProof
\end{adjustbox}
\end{center}

The case of $(\neg\iota I_2)$ is analogous. 
%
%
%
%
%
%
%
%
%
Next, consider $(\neg\iota E_2)$. If the term $t$ in the major premise $\iota xF\neq t$ is complex, introduce a fresh parameter $a$. The original application is replaced by the following:
\begin{center}
\begin{adjustbox}{max width=\textwidth}
\AxiomC{$\D$}
\noLine
\UnaryInfC{$\iota xF\neq t$}
\RightLabel{$(\neg\iota E_1)$}
\UnaryInfC{$\exists!t$}
\AxiomC{$[a=t]^i$}
\LL{$(sym)$}
\UnaryInfC{$t=a$}
\AxiomC{$\D$}
\noLine
\UnaryInfC{$\iota xF\neq t$}
\LL{$(=E')$}
\BinaryInfC{$\iota xF\neq a$}
\AxiomC{$[\exists!b]^j\;[F^x_b]^k\;[a\neq b]^\ell$}
\noLine
\UnaryInfC{$\D_1{}^t_a$}
\noLine
\UnaryInfC{$C$}
\AxiomC{$[\exists!c]^m\;[\neg F^x_c]^n\;[a=c]^p$}
\noLine
\UnaryInfC{$\D_2{}^t_a$}
\noLine
\UnaryInfC{$C$}
\LL{$(\neg\iota E_2)^{j,k,\ell,m,n,p}$}
\TrinaryInfC{$C$}
\RightLabel{$(=I^{nG})^i$}
\BinaryInfC{$C$}
\DisplayProof
\end{adjustbox}
\end{center}

$\D_1{}^t_a$ and $\D_2{}^t_a$ are obtained from the original
minor deductions by replacing $t$ by $a$, inserting applications of $(=E)$, $(=E')$, and $(\neq I)$ where required.
\end{proof}

\begin{definition}
Let $\KNFia$ be the system obtained from $\KNFi$ by replacing $(=I^n)$ with $(=I^{nG})$ and by restricting the instantiating terms in $(\forall E), (\exists I), (\neg\exists E), (\neg\forall I)$, and the relevant terms in $(\iota I), (\iota E_1), $ $(\iota E_2), $ $(\neg\iota I_1), (\neg\iota I_2), (\neg\iota E_2), (\neg\exists!\iota E)$, to atomic terms. 
\end{definition}

In what follows, we prove normalisation for $\KNFia$. By Lemmas \ref{InGder}, \ref{Ider}, \ref{restr:qua}, \ref{restr:negqua}, \ref{restr:iota}, and \ref{restr:negiota}, we get the following corollary.   
\begin{coro}
For each set of formulas $\Gamma\cup \{A\}$,   $\Gamma \vdash_{\KNFia}A$ iff $\Gamma \vdash_{\KNFi}A$.
\end{coro}

As pointed out in \cite[p. 260-261]{KurbisNormNeg}, K\"urbis \cite{KurbisIntuitionistic21,KurbisClassical22}
counts as maximal also those formulas which are discharged by general
introduction rules and subsequently occur as major premises of  elimination rules. For such formulas one could define a reduction procedure. In the case of an application of $(=I^{nG})$,
all maximal formulas generated by that application would have to be replaced simultaneously as follows:
\begin{center}
\begin{adjustbox}{max width=\textwidth}
\begin{tabular}{c@{\quad$\leadsto$\quad}c}
\AxiomC{$[a=t]^i$}
\AxiomC{$\mathfrak{E}$}
\noLine
\UnaryInfC{$A^x_a$}
\RightLabel{$(=E)$}
\BinaryInfC{$A^x_t$}
\noLine
\UnaryInfC{$\D$}
\noLine
\UnaryInfC{$C$}
\AxiomC{$\exists!t$}
\LL{$(=I^{nG})^i$}
\BinaryInfC{$C$}
\DisplayProof

&

\AxiomC{$\mathfrak{E}^a_t{}^*$}
\noLine
\UnaryInfC{$A^x_t$}
\noLine
\UnaryInfC{$\D^a_t{}^*$}
\noLine
\UnaryInfC{$C$}
\DisplayProof
\end{tabular}
\end{adjustbox}
\end{center}

If the assumption class of $a=t$ contains non-maximal formulas, then
$\D^a_t{}^*$ and $\mathfrak{E}^a_t{}^*$ are obtained by deleting the vacuous applications of $(=E)$ that arise after replacing $a$ by $t$. In the presence of $\iota$-terms, however, replacing a parameter by an arbitrary term may create maximal formulas of unknown degree. Since $t$ in $(=I^{nG})$ need not be atomic, the reduction above is not available in general. Following K\"urbis \cite[p.~261]{KurbisNormNeg}, we therefore do not count formulas discharged by $(=I^{nG})$ as maximal.

Let us describe the reduction procedures for maximal formulas. The usual permutation conversions for maximal segments generated by gel-rules are those of Prawitz and K\"{u}rbis and are included among our reduction procedures. The reduction procedures for $\wedge$, $\vee$, $\to$ are standard and described by Prawitz \cite{Prawitz}.   The reductions for $\forall$, $\exists$, and $\iota$ are presented by K\"{u}rbis \cite{KurbisNormNeg}. The reductions for $\neg\neg$, $\neg\wedge$, $\neg\vee$, and $\neg{\to}$ are given in \cite{KurbisPetrukhin21}. The reductions for $\neg\exists$, $\neg\forall$, negative $\iota$, and $\neg\exists!(\iota xF)$ are displayed below. 

\begin{center}
\begin{adjustbox}{max width=\textwidth}
\begin{tabular}{c@{\quad$\leadsto$\quad}c@{\qquad}c@{\quad$\leadsto$\quad}c}
\multicolumn{2}{c}{
$(\neg\exists I)$ followed by $(\neg\exists E)$}
&
\multicolumn{2}{c}{
$(\neg\forall I)$ followed by $(\neg\forall E)$}
\\[1.5ex]

\AxiomC{$[\exists!a]^i$}
\noLine
\UnaryInfC{$\D$}
\noLine
\UnaryInfC{$\neg A^x_a$}
\LL{$(\neg\exists I)^i$}
\UnaryInfC{$\neg\exists xA$}
\AxiomC{$\D'$}
\noLine
\UnaryInfC{$\exists!t$}
\LL{$(\neg\exists E)$}
\BinaryInfC{$\neg A^x_t$}
\DisplayProof

&

\AxiomC{$\D'$}
\noLine
\UnaryInfC{$[\exists!t]$}
\noLine
\UnaryInfC{$\D^a_t$}
\noLine
\UnaryInfC{$\neg A^x_t$}
\DisplayProof

&

\AxiomC{$\D'$}
\noLine
\UnaryInfC{$\neg A^x_t$}
\AxiomC{$\D''$}
\noLine
\UnaryInfC{$\exists!t$}
\LL{$(\neg\forall I)$}
\BinaryInfC{$\neg\forall xA$}
\AxiomC{$[\neg A^x_a]^i \; [\exists!a]^j$}
\noLine
\UnaryInfC{$\D$}
\noLine
\UnaryInfC{$C$}
\LL{$(\neg\forall E)^{i,j}$}
\BinaryInfC{$C$}
\DisplayProof

&

\AxiomC{$\D'\quad\; \D''$}
\noLine
\UnaryInfC{$[\neg A^x_t]\;[\exists!t]$}
\noLine
\UIC{$\D^a_t$}
\noLine
\UnaryInfC{$C$}
\DisplayProof
\end{tabular}
\end{adjustbox}
\end{center}

\begin{center}
\begin{adjustbox}{max width=\textwidth}
\begin{tabular}{c@{\quad$\leadsto$\quad}c@{\qquad}c@{\quad$\leadsto$\quad}c}
\multicolumn{2}{c}{
$(\neg\iota I_1)$ followed by $(\neg\iota E_1)$}
&
\multicolumn{2}{c}{
$(\neg\exists!\iota I)$ followed by $(\neg\exists!\iota E)$}
\\[1.5ex]

\AxiomC{$\D_1$}
\noLine
\UnaryInfC{$\exists!t$}
\AxiomC{$\D_2$}
\noLine
\UnaryInfC{$\exists!u$}
\AxiomC{$\D_3$}
\noLine
\UnaryInfC{$F^x_u$}
\AxiomC{$\D_4$}
\noLine
\UnaryInfC{$t\neq u$}
\LL{$(\neg\iota I_1)$}
\QuaternaryInfC{$\iota xF\neq t$}
\LL{$(\neg\iota E_1)$}
\UnaryInfC{$\exists!t$}
\DisplayProof

&

\AxiomC{$\D_1$}
\noLine
\UnaryInfC{$\exists!t$}
\DisplayProof

&

\AxiomC{$[\exists!a]^i$}
\noLine
\UnaryInfC{$\D$}
\noLine
\UnaryInfC{$\iota xF\neq a$}
\LL{$(\neg\exists!\iota I)^i$}
\UnaryInfC{$\neg\exists!(\iota xF)$}

\AxiomC{$\D'$}
\noLine
\UnaryInfC{$\exists!t$}
\LL{$(\neg\exists!\iota E)$}
\BinaryInfC{$\iota xF\neq t$}
\DisplayProof

&

\AxiomC{$\D'$}
\noLine
\UnaryInfC{$[\exists!t]$}
\noLine
\UnaryInfC{$\D^a_t$}
\noLine
\UnaryInfC{$\iota xF\neq t$}
\DisplayProof

\end{tabular}
\end{adjustbox}
\end{center}

($\neg \iota I_2$) followed by ($\neg \iota E_1$). Similarly to the previous one.
%
%
%
%
%
%
%
%
%
%
($\neg \iota I_1$) followed by ($\neg \iota E_2$)
\begin{center}
\begin{adjustbox}{max width=\textwidth}
\begin{tabular}{c@{\quad$\leadsto$\quad}c}
\AxiomC{$\D_1$}
\noLine
\UnaryInfC{$\exists!t$}
\AxiomC{$\D_2$}
\noLine
\UnaryInfC{$\exists!u$}
\AxiomC{$\D_3$}
\noLine
\UnaryInfC{$F^x_u$}
\AxiomC{$\D_4$}
\noLine
\UnaryInfC{$t\neq u$}
\LeftLabel{$(\neg\iota I_1)$}
\QuaternaryInfC{$\iota xF\neq t$}
\AxiomC{$[\exists!a]^i\,[F^x_a]^j\,[t\neq a]^k$}
\noLine
\UnaryInfC{$\D$}
\noLine
\UnaryInfC{$C$}
\AxiomC{$[\exists!b]^\ell\,[\neg F^x_b]^m\,[t=b]^n$}
\noLine
\UnaryInfC{$\D'$}
\noLine
\UnaryInfC{$C$}
\LeftLabel{$(\neg\iota E_2)^{i,\ldots,n}$}
\TrinaryInfC{$C$}
\DisplayProof

&

\AxiomC{$\D_2$}
\noLine
\UnaryInfC{$[\exists!u]$}
\AxiomC{$\D_3$}
\noLine
\UnaryInfC{$[F^x_u]$}
\AxiomC{$\D_4$}
\noLine
\UnaryInfC{$[t\neq u]$}
\noLine
\TrinaryInfC{$\D^a_u$}
\noLine
\UnaryInfC{$C$}
\DisplayProof
\end{tabular}
\end{adjustbox}
\end{center}

($\neg \iota I_2$) followed by ($\neg \iota E_2$). Similarly to the previous case.
%
%
%
($\not = I$) followed by ($\neg \iota E_2$).
\begin{center}
\begin{adjustbox}{max width=\textwidth}
\begin{tabular}{c@{\quad$\leadsto$\quad}c}
\AxiomC{$\D_1$}
\noLine
\UnaryInfC{$A^x_{\iota xF}$}
\AxiomC{$\D_2$}
\noLine
\UnaryInfC{$\neg A^x_t$}
\LL{$(\neq I)$}
\BinaryInfC{$\iota xF\neq t$}
\AxiomC{$[\exists!a]^i\;[F^x_a]^j\;[t\neq a]^k$}
\noLine
\UnaryInfC{$\D_3$}
\noLine
\UnaryInfC{$C$}
\AxiomC{$[\exists!b]^\ell\;[\neg F^x_b]^m\;[t=b]^n$}
\noLine
\UnaryInfC{$\D_4$}
\noLine
\UnaryInfC{$C$}
\LL{$(\neg\iota E_2)^{i,\ldots,n}$}
\TrinaryInfC{$C$}
\DisplayProof

&

\AxiomC{$\D_1$}
\noLine
\UnaryInfC{$A^x_{\iota xF}$}
\LL{$(AD)$}
\UnaryInfC{$\exists!(\iota xF)$}
\AxiomC{$[a=\iota xF]^r$}
\noLine
\UnaryInfC{$\D_3^*$}
\noLine
\UnaryInfC{$C$}
\LL{$(=I^{nG})^r$}
\BinaryInfC{$C$}
\DisplayProof
\end{tabular}
\end{adjustbox}
\end{center}

$\D_3^*$ is obtained from $\D_3$ by replacing the assumptions
$[\exists!a]^i$, $[F^x_a]^j$, and $[t\neq a]^k$ by the deductions:
\begin{center}
\begin{adjustbox}{max width=\textwidth}
\begin{tabular}{c@{\qquad}c@{\qquad}c}
\AxiomC{$[a=\iota xF]^r$}
\RightLabel{$(AD)$}
\UnaryInfC{$\exists!a$}
\DisplayProof

&

\AxiomC{$[a=\iota xF]^r$}
\RightLabel{$(sym)$}
\UnaryInfC{$\iota xF=a$}
\AxiomC{$[a=\iota xF]^r$}
\RightLabel{$(AD)$}
\UnaryInfC{$\exists!a$}
\RightLabel{$(=I^n)$}
\UnaryInfC{$a=a$}
\RightLabel{$(\iota E_1)$}
\BinaryInfC{$F^x_a$}
\DisplayProof

&

\AxiomC{$\D_5$}
\noLine
\UnaryInfC{$a\neq t$}
\RightLabel{$(\not=D)$}
\UnaryInfC{$\exists!t$}
\RightLabel{$(=I^n)$}
\UnaryInfC{$t=t$}
\AxiomC{$\D_5$}
\noLine
\UnaryInfC{$a\neq t$}
\RightLabel{$(\neq I)$}
\BinaryInfC{$t\neq a$}
\DisplayProof
\end{tabular}
\end{adjustbox}
\end{center}

where $\D_5$ is as follows:
\begin{center}
\AxiomC{$\D_1$}
\noLine
\UnaryInfC{$A^x_{\iota xF}$}
\AxiomC{$[a=\iota xF]^r$}
\RightLabel{$(sym)$}
\UnaryInfC{$\iota xF=a$}
\RightLabel{$(=E)$}
\BinaryInfC{$A^x_a$}
\AxiomC{$\D_2$}
\noLine
\UnaryInfC{$\neg A^x_t$}
\RightLabel{$(\neq I)$}
\BinaryInfC{$a\neq t$}
\DisplayProof
\end{center}

We also need special reductions involving the rule (\textit{EFQ}).
The corresponding reductions for (\textit{EFQ}) and the propositional
negation rules are given in \cite{KurbisPetrukhin21}. As an example,
let us consider the case of (\textit{EFQ}) following
$(\neg\forall I)$. The cases involving $(\neg\exists I)$, $(\neq I)$,
$(\neg\iota I_1)$, $(\neg\iota I_2)$, and
$(\neg\exists!\iota I)$ are treated similarly.
\begin{center}
\begin{adjustbox}{max width=\textwidth}
\begin{tabular}{c@{\qquad$\leadsto$\qquad}c}
\AxiomC{$\D_1$}
\noLine
\UnaryInfC{$\neg A^x_t$}
\AxiomC{$\D_2$}
\noLine
\UnaryInfC{$\exists!t$}
\LL{$(\neg\forall I)$}
\BinaryInfC{$\neg\forall xA$}
\AxiomC{$\D_3$}
\noLine
\UnaryInfC{$\forall xA$}
\LL{(\textit{EFQ})}
\BinaryInfC{$C$}
\DisplayProof

&

\AxiomC{$\D_1$}
\noLine
\UnaryInfC{$\neg A^x_t$}
\AxiomC{$\D_3$}
\noLine
\UnaryInfC{$\forall xA$}
\AxiomC{$\D_2$}
\noLine
\UnaryInfC{$\exists!t$}
\LL{$(\forall E)$}
\BinaryInfC{$A^x_t$}
\LL{(\textit{EFQ})}
\BinaryInfC{$C$}
\DisplayProof
\end{tabular}
\end{adjustbox}
\end{center}

\begin{definition}
A \emph{path} in a deduction is a sequence of formula occurrences
$A_1,\ldots,A_n$ such that $A_1$ is an assumption, $A_n$ is the conclusion of the deduction, and, for each $i<n$, $A_i$ is a premise immediately above $A_{i+1}$. For applications of $(=I^{nG})$, paths continue from the formula $C$ immediately above the rule to the conclusion $C$ of that application. They also continue, when
needed, to assumptions discharged by that application. 
A \emph{track} is an initial part of a path which stops either at the first minor premise of an elimination rule or at the conclusion of the deduction; in the presence of general elimination rules, a track leads as usual from the major premise through the cancelled assumptions downwards to the conclusion.
\end{definition}

\begin{definition}\cite[p. 565]{Zimmermann}
A \emph{trace} is a part of a track, where the first formula occurrence is a conclusion of an application of Peirce's rule $(P)$,  and every formula occurrence is a conclusion of an arbitrary elimination rule. A trace is non-trivial if its length is greater than $1$. A deduction is $P$-normal if it contains no non-trivial traces of applications of $(P)$.
\end{definition}

Partial implication contractions reduce the length of traces. They do not remove the application of $(P)$ altogether, but replace it by a new application of $(P)$ whose trace is shorter. The cases of $(P)$ followed by $(\vee E)$, $(\to E)$, $(\wedge E)$ are covered in \cite{Zimmermann}. The cases of $(P)$ followed by $(\neg\neg E)$, $(\neg{\vee} E_1)$, $(\neg{\vee} E_2)$, $(\neg{\to} E_1)$, $(\neg{\to} E_2)$, $(\neg{\wedge} E)$ are described in \cite{KurbisPetrukhin21}. 
Let us describe the remaining cases. 

$(P)$ followed by $(\forall E)$
\begin{center}
\begin{adjustbox}{max width=\textwidth}
\begin{tabular}{c@{\quad$\leadsto$\quad}c}
\AxiomC{$[\forall xA\to C]^i$}
\noLine
\UnaryInfC{$\D$}
\noLine
\UnaryInfC{$\forall xA$}
\LL{$(P)^i$}
\UnaryInfC{$\forall xA$}
\AxiomC{$\D'$}
\noLine
\UnaryInfC{$\exists!t$}
\LL{$(\forall E)$}
\BinaryInfC{$A^x_t$}
\DisplayProof

&

\AxiomC{$[A^x_t\to C]^j$}
\AxiomC{$[\forall xA]^k$}
\AxiomC{$\D'$}
\noLine
\UnaryInfC{$\exists!t$}
\LL{$(\forall E)$}
\BinaryInfC{$A^x_t$}
\RL{$(\to E)$}
\BinaryInfC{$C$}
\LL{$(\to I)^k$}
\UnaryInfC{$\forall xA\to C$}
\noLine
\UnaryInfC{$\D$}
\noLine
\UnaryInfC{$\forall xA$}
\AxiomC{$\D'$}
\noLine
\UnaryInfC{$\exists!t$}
\LL{$(\forall E)$}
\BinaryInfC{$A^x_t$}
\LL{$(P)^j$}
\UnaryInfC{$A^x_t$}
\DisplayProof
\end{tabular}
\end{adjustbox}
\end{center}

$(P)$ followed by $(\exists E)$
\begin{center}
\begin{adjustbox}{max width=\textwidth}
\begin{tabular}{c@{\quad$\leadsto$\quad}c}
\AxiomC{$[\exists xA\to C]^i$}
\noLine
\UnaryInfC{$\D$}
\noLine
\UnaryInfC{$\exists xA$}
\LL{$(P)^i$}
\UnaryInfC{$\exists xA$}
\AxiomC{$[A^x_a]^k\, [\exists!a]^\ell$}
\noLine
\UnaryInfC{$\D'$}
\noLine
\UnaryInfC{$B$}
\LL{$(\exists E)^{k,\ell}$}
\BinaryInfC{$B$}
\DisplayProof

&

\AxiomC{$[B\to C]^j$}
\AxiomC{$[\exists xA]^m$}
\AxiomC{$[A^x_a]^k\, [\exists!a]^\ell$}
\noLine
\UnaryInfC{$\D'$}
\noLine
\UnaryInfC{$B$}
\LL{$(\exists E)^{k,\ell}$}
\BinaryInfC{$B$}
\RL{$(\to E)$}
\BinaryInfC{$C$}
\LL{$(\to I)^m$}
\UnaryInfC{$\exists xA\to C$}
\noLine
\UnaryInfC{$\D$}
\noLine
\UnaryInfC{$\exists xA$}
\AxiomC{$[A^x_a]^k\, [\exists!a]^\ell$}
\noLine
\UnaryInfC{$\D'$}
\noLine
\UnaryInfC{$B$}
\LL{$(\exists E)^{k,\ell}$}
\BinaryInfC{$B$}
\LL{$(P)^j$}
\UnaryInfC{$B$}
\DisplayProof
\end{tabular}
\end{adjustbox}
\end{center}

$(P)$ followed by $(\neg\forall E)$. 
\begin{center}
\begin{adjustbox}{max width=\textwidth}
\begin{tabular}{c@{\quad$\leadsto$\quad}c}

\AxiomC{$[\neg\forall xA\to C]^i$}
\noLine
\UnaryInfC{$\D$}
\noLine
\UnaryInfC{$\neg\forall xA$}
\LL{$(P)^i$}
\UnaryInfC{$\neg\forall xA$}

\AxiomC{$[\neg A^x_a]^k\, [\exists!a]^\ell$}
\noLine
\UnaryInfC{$\D'$}
\noLine
\UnaryInfC{$B$}

\LL{$(\neg\forall E)^{k,\ell}$}
\BinaryInfC{$B$}
\DisplayProof

&

\AxiomC{$[B\to C]^j$}

\AxiomC{$[\neg\forall xA]^m$}
\AxiomC{$[\neg A^x_a]^k\, [\exists!a]^\ell$}
\noLine
\UnaryInfC{$\D'$}
\noLine
\UnaryInfC{$B$}
\RL{$(\neg\forall E)^{k,\ell}$}
\BinaryInfC{$B$}

\RL{$(\to E)$}
\BinaryInfC{$C$}
\RL{$(\to I)^m$}
\UnaryInfC{$\neg\forall xA\to C$}

\noLine
\UnaryInfC{$\D$}
\noLine
\UnaryInfC{$\neg\forall xA$}

\AxiomC{$[\neg A^x_a]^k\, [\exists!a]^\ell$}
\noLine
\UnaryInfC{$\D'$}
\noLine
\UnaryInfC{$B$}

\LL{$(\neg\forall E)^{k,\ell}$}
\BinaryInfC{$B$}

\LL{$(P)^j$}
\UnaryInfC{$B$}
\DisplayProof

\end{tabular}
\end{adjustbox}
\end{center}

$(P)$ followed by $(\neg\exists E)$
\begin{center}
\begin{adjustbox}{max width=\textwidth}
\begin{tabular}{c@{\quad$\leadsto$\quad}c}
\AxiomC{$[\neg\exists xA\to C]^i$}
\noLine
\UnaryInfC{$\D$}
\noLine
\UnaryInfC{$\neg\exists xA$}
\LL{$(P)^i$}
\UnaryInfC{$\neg\exists xA$}

\AxiomC{$\D'$}
\noLine
\UnaryInfC{$\exists!t$}

\LL{$(\neg\exists E)$}
\BinaryInfC{$\neg A^x_t$}
\DisplayProof

&

\AxiomC{$[\neg A^x_t\to C]^j$}

\AxiomC{$[\neg\exists xA]^k$}
\AxiomC{$\D'$}
\noLine
\UnaryInfC{$\exists!t$}
\LL{$(\neg\exists E)$}
\BinaryInfC{$\neg A^x_t$}

\RL{$(\to E)$}
\BinaryInfC{$C$}
\LL{$(\to I)^k$}
\UnaryInfC{$\neg\exists xA\to C$}

\noLine
\UnaryInfC{$\D$}
\noLine
\UnaryInfC{$\neg\exists xA$}

\AxiomC{$\D'$}
\noLine
\UnaryInfC{$\exists!t$}

\LL{$(\neg\exists E)$}
\BinaryInfC{$\neg A^x_t$}

\LL{$(P)^j$}
\UnaryInfC{$\neg A^x_t$}
\DisplayProof

\end{tabular}
\end{adjustbox}
\end{center}

 $(P)$ followed by $(\iota E_1)$
\begin{center}
\begin{adjustbox}{max width=\textwidth}
\begin{tabular}{c@{\quad$\leadsto$\quad}c}

\AxiomC{$[\iota xF=t\to C]^i$}
\noLine
\UnaryInfC{$\D$}
\noLine
\UnaryInfC{$\iota xF=t$}
\LL{$(P)^i$}
\UnaryInfC{$\iota xF=t$}

\AxiomC{$\D'$}
\noLine
\UnaryInfC{$u=t$}

\LL{$(\iota E_1)$}
\BinaryInfC{$F^x_u$}
\DisplayProof

&

\AxiomC{$[F^x_u\to C]^j$}

\AxiomC{$[\iota xF=t]^k$}
\AxiomC{$\D'$}
\noLine
\UnaryInfC{$u=t$}
\LL{$(\iota E_1)$}
\BinaryInfC{$F^x_u$}

\RL{$(\to E)$}
\BinaryInfC{$C$}
\LL{$(\to I)^k$}
\UnaryInfC{$\iota xF=t\to C$}

\noLine
\UnaryInfC{$\D$}
\noLine
\UnaryInfC{$\iota xF=t$}

\AxiomC{$\D'$}
\noLine
\UnaryInfC{$u=t$}

\LL{$(\iota E_1)$}
\BinaryInfC{$F^x_u$}

\LL{$(P)^j$}
\UnaryInfC{$F^x_u$}
\DisplayProof

\end{tabular}
\end{adjustbox}
\end{center}

$(P)$ followed by $(\iota E_2)$. 
%
%
%
%
%
%
%
%
%
%
%
%
%
%
$(P)$ followed by $(\iota E_3)$. 
%
%
%
%
%
%
%
%
%
%
$(P)$ followed by $(\neg\iota E_1)$. Similarly to the previous cases. 
$(P)$ followed by $(\neg\iota E_2)$
\begin{center}
    \AxiomC{$[\iota xF\neq t\to C]^r$}
\noLine
\UnaryInfC{$\D_1$}
\noLine
\UnaryInfC{$\iota xF\neq t$}
\RL{$(P)^r$}
\UnaryInfC{$\iota xF\neq t$}

\AxiomC{$[\exists!a]^i[F^x_a]^j[t\neq a]^k$}
\noLine
\UnaryInfC{$\D_2$}
\noLine
\UnaryInfC{$B$}

\AxiomC{$[\exists!b]^\ell[\neg F^x_b]^m[t=b]^n$}
\noLine
\UnaryInfC{$\D_3$}
\noLine
\UnaryInfC{$B$}

\LL{$(\neg\iota E_2)^{i,j,k,\ell,m,n}$}
\TrinaryInfC{$B$}
\DisplayProof
\quad$\leadsto$\quad
\end{center}

\begin{center}
\begin{adjustbox}{max width=\textwidth}
    \AxiomC{$[B\to C]^s$}
\AxiomC{$[\iota xF\neq t]^q$}
\AxiomC{$[\exists!a]^i[F^x_a]^j[t\neq a]^k$}
\noLine
\UnaryInfC{$\D_2$}
\noLine
\UnaryInfC{$B$}
\AxiomC{$[\exists!b]^\ell[\neg F^x_b]^m[t=b]^n$}
\noLine
\UnaryInfC{$\D_3$}
\noLine
\UnaryInfC{$B$}
\LL{$(\neg\iota E_2)^{i,j,k,\ell,m,n}$}
\TrinaryInfC{$B$}
\RL{$(\to E)$}
\BinaryInfC{$C$}
\RL{$(\to I)^q$}
\UnaryInfC{$\iota xF\neq t\to C$}
\noLine
\UnaryInfC{$\D_1$}
\noLine
\UnaryInfC{$\iota xF\neq t$}
\AxiomC{$[\exists!a]^i[F^x_a]^j[t\neq a]^k$}
\noLine
\UnaryInfC{$\D_2$}
\noLine
\UnaryInfC{$B$}
\AxiomC{$[\exists!b]^\ell[\neg F^x_b]^m[t=b]^n$}
\noLine
\UnaryInfC{$\D_3$}
\noLine
\UnaryInfC{$B$}
\LL{$(\neg\iota E_2)^{i,j,k,\ell,m,n}$}
\TrinaryInfC{$B$}
\LL{$(P)^s$}
\UnaryInfC{$B$}
\DisplayProof
\end{adjustbox}
\end{center}

$(P)$ followed by $(\neg\exists!\iota E)$
\begin{center}
\begin{adjustbox}{max width=\textwidth}
\begin{tabular}{c@{\quad$\leadsto$\quad}c}

\AxiomC{$[\neg\exists!(\iota xF)\to C]^i$}
\noLine
\UnaryInfC{$\D$}
\noLine
\UnaryInfC{$\neg\exists!(\iota xF)$}
\LL{$(P)^i$}
\UnaryInfC{$\neg\exists!(\iota xF)$}

\AxiomC{$\D'$}
\noLine
\UnaryInfC{$\exists!t$}

\LL{$(\neg\exists!\iota E)$}
\BinaryInfC{$\iota xF\neq t$}
\DisplayProof

&

\AxiomC{$[\iota xF\neq t\to C]^j$}

\AxiomC{$[\neg\exists!(\iota xF)]^k$}
\AxiomC{$\D'$}
\noLine
\UnaryInfC{$\exists!t$}
\LL{$(\neg\exists!\iota E)$}
\BinaryInfC{$\iota xF\neq t$}

\RL{$(\to E)$}
\BinaryInfC{$C$}
\LL{$(\to I)^k$}
\UnaryInfC{$\neg\exists!(\iota xF)\to C$}

\noLine
\UnaryInfC{$\D$}
\noLine
\UnaryInfC{$\neg\exists!(\iota xF)$}

\AxiomC{$\D'$}
\noLine
\UnaryInfC{$\exists!t$}

\LL{$(\neg\exists!\iota E)$}
\BinaryInfC{$\iota xF\neq t$}

\LL{$(P)^j$}
\UnaryInfC{$\iota xF\neq t$}
\DisplayProof

\end{tabular}
\end{adjustbox}
\end{center}

$(P)$ followed by $(=E)$. The case $(P)$ followed by $(=E')$ is treated similarly.
\begin{center}
\begin{adjustbox}{max width=\textwidth}
\begin{tabular}{c@{\quad$\leadsto$\quad}c}

\AxiomC{$[t_1=t_2\to C]^i$}
\noLine
\UnaryInfC{$\D$}
\noLine
\UnaryInfC{$t_1=t_2$}
\LL{$(P)^i$}
\UnaryInfC{$t_1=t_2$}

\AxiomC{$\D'$}
\noLine
\UnaryInfC{$\mathcal A^x_{t_1}$}

\LL{$(=E)$}
\BinaryInfC{$\mathcal A^x_{t_2}$}
\DisplayProof

&

\AxiomC{$[\mathcal A^x_{t_2}\to C]^j$}

\AxiomC{$[t_1=t_2]^k$}

\AxiomC{$\D'$}
\noLine
\UnaryInfC{$\mathcal A^x_{t_1}$}

\LL{$(=E)$}
\BinaryInfC{$\mathcal A^x_{t_2}$}

\RL{$(\to E)$}
\BinaryInfC{$C$}

\LL{$(\to I)^k$}
\UnaryInfC{$t_1=t_2\to C$}

\noLine
\UnaryInfC{$\D$}
\noLine
\UnaryInfC{$t_1=t_2$}

\AxiomC{$\D'$}
\noLine
\UnaryInfC{$\mathcal A^x_{t_1}$}

\LL{$(=E)$}
\BinaryInfC{$\mathcal A^x_{t_2}$}

\LL{$(P)^j$}
\UnaryInfC{$\mathcal A^x_{t_2}$}
\DisplayProof

\end{tabular}
\end{adjustbox}
\end{center}

%
%
%
%
%
%
%
%
%
%
%
%
%
%
%

\begin{lemma}\label{lem:Pnormal}
Every deduction in $\KNFia$ can be transformed into a $P$-normal deduction.
\end{lemma}
\begin{proof}
As in \cite[Lemma~5]{Zimmermann}, an induction on the length of traces and the degree of their $(P)$-conclusions shows that successive partial implication contractions eliminate all non-trivial traces.
\end{proof}

\begin{definition}\cite[Definition~19]{KurbisNormNeg}
A $\exists!$-\emph{segment} is a segment all formulas of which are
occurrences of the formula $\exists!t$. An $\exists!t$-segment is
\emph{maximal} if its first formula is concluded by $(AD)$ or
$(\iota E_3)$ and its last formula is the $\exists!t$-premise of an
application of $(=I^{nG})$ with conclusion $t=t$.
\end{definition}

\begin{definition}\cite[Definition~10]{KurbisNormNeg}
An $(=E)$-\emph{segment} is a sequence of segments
$\sigma_1,\ldots,\sigma_n$ such that $\sigma_1$ is the minor premise
but not the conclusion of an application of $(=E)$, $\sigma_n$ is the
conclusion but not the minor premise of an application of $(=E)$, and, if $1<i<n$, then $\sigma_i$ is both the conclusion of an application of $(=E)$ and the minor premise of an application of $(=E)$. An $(=E)$-segment is \emph{maximal} if its last segment is a premise of $(AD)$ or $(\iota E_3)$, or if all formulas on it are of the form $\exists!t$.
\end{definition}

\begin{definition}
An $(=E')$-segment is defined analogously to an $(=E)$-segment, with
$(=E')$ in place of $(=E)$. It is maximal if its last segment is a
premise of $(AD')$, $(\neq D)$, or $(\neg\iota E_1)$.
\end{definition}

\begin{definition}\cite[Definitions~9 and~19]{KurbisNormNeg}
A maximal $=$-\emph{segment} is either (i)
 a segment of formulas $t=t$ such that one of its formulas is
the minor premise of an application of $(=I^{nG})$ and its last formula is a premise of $(AD)$ or $(\iota E_3)$; or (ii) a segment of formulas $a=t$ such that its first formula is
discharged by an application of $(=I^{nG})$ and its last formula is a
premise of $(AD)$ or $(\iota E_3)$ with conclusion $\exists!t$.
\end{definition}

\begin{definition}
A maximal $\neq$-\emph{segment} is a segment of formulas of the form
$t_1\neq t_2$ satisfying one of the following conditions: (i) its first formula is concluded by $(\neq I)$ and its last
formula is a premise of $(\neq D)$; (ii) its first formula is concluded by $(\neg\iota I_1)$,
$(\neg\iota I_2)$, or $(\neg\exists!\iota E)$ and its last formula is
a premise of $(\neq D)$ or $(\neg\iota E_1)$.
\end{definition}

\begin{lemma}\label{lem:specialsegments}
Any deduction in $\KNFia$ can be transformed into one without maximal
$\exists!$-, $=$-, $\neq$-, $(=E)$-, or $(=E')$-segments.
\end{lemma}

\begin{proof}
For maximal $\exists!$-, $=$-, and $(=E)$-segments we use K\"urbis'
procedures \cite[Lemmas~3 and~9]{KurbisNormNeg}; vacuous applications of $(=E)$ and $(=E')$ are deleted whenever they arise. The new negative cases are analogous. A maximal $\neq$-segment of type~(i) ending in $(\neq D)$ is bypassed: $\exists!t_1$ follows from $\A^x_{t_1}$ by
$(AD)$, while $\exists!t_2$ follows from $\neg\A^x_{t_2}$ by $(AD')$
if the principal symbol of $\A$ is an ordinary predicate letter, and by $(\neq D)$ if it is $=$, since in the latter case $t_2$ is atomic. In type~(ii), the required $\exists!t$ is
already a premise of $(\neg\iota I_1)$, $(\neg\iota I_2)$, or
$(\neg\exists!\iota E)$, so the segment and the final $(\neq D)$ or
$(\neg\iota E_1)$ are bypassed. Maximal $(=E')$-segments are treated as maximal $(=E)$-segments, proceeding directly to $(AD')$, $(\neq D)$, or $(\neg\iota E_1)$. For termination, let $c(\D)$ be the number of maximal $=$-segments of
type~(ii), and $s(\D)$ the number of rule applications, and order
$\mu(\D)=\langle c(\D),s(\D)\rangle$ lexicographically. Every reduction except that of a maximal $=$-segment of type~(ii) strictly decreases $s(\D)$ and creates no segment of type~(ii). In the exceptional case, choose such a segment with none above it and first remove all special maximal segments above it. Its reduction may increase $s(\D)$, but no segment of type~(ii) is duplicated and the selected one is removed, so $c(\D)$ strictly decreases. Hence $\mu(\D)$ strictly decreases and the procedure terminates.
\end{proof}
\begin{definition}\label{rank}
The \emph{rank} of a deduction $\D$ is the pair
$r(\D)=\langle d,l\rangle$, ordered lexicographically, where $d$ is
the greatest degree of a maximal formula or maximal segment
in $\D$, and $l$ is the sum of their lengths at degree $d$, maximal
formulas being counted as segments of length~$1$. If there are none,
$r(\D)=\langle 0,0\rangle$. A deduction is in \emph{normal form} if it is $P$-normal and contains no maximal formulas, maximal segments, or special maximal segments of Lemma~\ref{lem:specialsegments}.
\end{definition}

\begin{theorem}\label{thm:norm}
Deductions in $\KNFia$ can be brought into normal form.    
\end{theorem}
\begin{proof}
First apply Lemma~\ref{lem:Pnormal} and then 
Lemma~\ref{lem:specialsegments}. Suppose that the resulting deduction is not normal. Choose a rightmost maximal formula or maximal segment of highest degree such that none of the same degree stands above it, and apply the corresponding reduction or permutation. As usual \cite{Prawitz,KurbisNormNeg}, and because the relevant instantiating terms are restricted to atomic terms, no maximal formula or maximal segment of higher degree is introduced. Any non-trivial $P$-trace created by the reduction is removed by
Lemma~\ref{lem:Pnormal}, and any newly created special segment is
removed by Lemma~\ref{lem:specialsegments}. By the choice of the reduced occurrence, these auxiliary reductions affect only material of lower degree and therefore do not increase the rank. Thus, the chosen maximal formula is removed or the chosen maximal segment is shortened, so $r(\D)$ strictly decreases. Repetition of this procedure terminates in a normal deduction.
\end{proof}

As in \cite{KurbisNormNeg}, substitution instances $A^x_t$ with atomic $t$ are counted as subformulas of $\forall xA$ and $\exists xA$. As in \cite{KurbisNormNeg}, we establish the subformula property for the $\iota$-free fragment of the system.
\begin{definition}
A deduction $\D$ of a conclusion $A$ from the undischarged assumptions $\Gamma$ has the \emph{Peirce negation free subformula property} if every formula occurring in $\D$ is either a subformula of $A$ or of a formula in $\Gamma$, a formula $B\to C$ discharged by $(P)$, a subformula of a formula $C$ occurring immediately below such a discharged formula, or the negation of one of these formulas, except possibly: formulas of the form $t_1=t_2$, $t_1\neq t_2$,
$Rt_1\ldots t_n$, or $\neg Rt_1\ldots t_n$ on $(=E)$- or
$(=E')$-segments; formulas of the form $\exists!t$ that are existence premises of $(\forall E)$, $(\exists I)$, $(\neg\exists E)$, or $(\neg\forall I)$, or premises of $(=I^n)$; or formulas of the form $\exists!a$ that are discharged by $(\forall I)$, $(\exists E)$, $(\neg\exists I)$, or $(\neg\forall E)$.
\end{definition}
\begin{theorem}
Every deduction in $\KNF$ can be transformed into a normal deduction with the Peirce negation free subformula property.
\end{theorem}

\begin{proof}
First restrict the conclusions of (\textit{EFQ}) to prime and negative prime formulas by Lemma~\ref{lem:efq}, and apply the normalisation procedure together with the (\textit{EFQ})-reductions described above. As in \cite[Theorem~21]{KurbisPetrukhin21}, these reductions terminate, since they lower the degree of the contradictory formula and any maximal formulas they create have lower degree. The subformula property then follows from the standard path arguments of \cite[Theorem~21]{KurbisPetrukhin21} and \cite[Corollary~5]{KurbisNormNeg}. The additional negative free-logical rules introduce only the exceptional formulas listed in the definition: $(=E')$ yields negative atomic formulas on $(=E')$-segments, while $(AD')$ and $(\neq D)$ may yield the $\exists!$-premises required by identity and quantifier rules.
\end{proof}

\section{Conclusion}\label{Concl}
An important task for future research is to extend the three-valued
semantics of the propositional fragment of $\KNFi$
\cite{SlupeckiPComp,avronnaturallogics,Kleene38} to quantifiers, the existence predicate, identity, and definite descriptions. Its implication-free propositional fragment is the Strong Kleene logic $\mathbf{K}_3$ \cite{Kleene38}. We conjecture that the corresponding first-order semantics can be developed in the spirit of Antunes and Rodrigues' $\mathbf{FFDE}$ \cite{AntunesRodrigues25} and the author's
$\mathbf{PRL}$ \cite{Petrukhin2026PRL}. Another task is to find a
suitable subformula property for $\KNFia$. Finally, a further direction is to investigate automated proof search for $\KNFi$,
in the spirit of the methods developed for $\mathbf{K}_3$ and its truth-functional extensions in \cite{PetrukhinShanginK3}.

\paragraph{Acknowledgements.}  
The author thanks anonymous reviewers for their useful and detailed comments. This research is funded by the European Union (ERC, ExtenDD, project number: 101054714). Views and opinions expressed are however those of the author(s) only and do not necessarily reflect those of the European Union or the European Research Council. Neither the European Union nor the granting authority can be held responsible for them. 

\bibliographystyle{eptcs}
\bibliography{references}

\begin{thebibliography}{10}
\providecommand{\bibitemdeclare}[2]{}
\providecommand{\surnamestart}{}
\providecommand{\surnameend}{}
\providecommand{\urlprefix}{Available at }
\providecommand{\url}[1]{\texttt{#1}}
\providecommand{\href}[2]{\texttt{#2}}
\providecommand{\urlalt}[2]{\href{#1}{#2}}
\providecommand{\doi}[1]{doi:\urlalt{https://doi.org/#1}{#1}}
\providecommand{\eprint}[1]{arXiv:\urlalt{https://arxiv.org/abs/#1}{#1}}
\providecommand{\bibinfo}[2]{#2}

\bibitemdeclare{article}{AlmukdadNelson84}
\bibitem{AlmukdadNelson84}
\bibinfo{author}{A.~\surnamestart Almukdad\surnameend} \& \bibinfo{author}{D.~\surnamestart Nelson\surnameend} (\bibinfo{year}{1984}): \emph{\bibinfo{title}{Constructible falsity and inexact predicates}}.
\newblock {\slshape \bibinfo{journal}{The Journal of Symbolic Logic}} \bibinfo{volume}{49}(\bibinfo{number}{1}), pp. \bibinfo{pages}{231--233}, \doi{10.2307/2274105}.

\bibitemdeclare{article}{AndersonBelnap63}
\bibitem{AndersonBelnap63}
\bibinfo{author}{A.~R. \surnamestart Anderson\surnameend} \& \bibinfo{author}{N.~D. \surnamestart Belnap\surnameend} (\bibinfo{year}{1963}): \emph{\bibinfo{title}{First degree entailments}}.
\newblock {\slshape \bibinfo{journal}{Mathematische Annalen}} \bibinfo{volume}{149}, pp. \bibinfo{pages}{302--319}, \doi{10.1007/BF01471125}.

\bibitemdeclare{article}{AntunesRodrigues25}
\bibitem{AntunesRodrigues25}
\bibinfo{author}{H.~\surnamestart Antunes\surnameend} \& \bibinfo{author}{A.~\surnamestart Rodrigues\surnameend} (\bibinfo{year}{2025}): \emph{\bibinfo{title}{On Universally Free First-Order Extensions of {Belnap-Dunn's} Four-Valued Logic and {Nelson's} Paraconsistent Logic {$N4$}}}.
\newblock {\slshape \bibinfo{journal}{Journal of Philosophical Logic}} \bibinfo{volume}{54}(\bibinfo{number}{1}), pp. \bibinfo{pages}{169--195}, \doi{10.1007/s10992-025-09783-w}.

\bibitemdeclare{article}{avronnaturallogics}
\bibitem{avronnaturallogics}
\bibinfo{author}{A.~\surnamestart Avron\surnameend} (\bibinfo{year}{1991}): \emph{\bibinfo{title}{Natural 3-valued logics: characterization and proof theory}}.
\newblock {\slshape \bibinfo{journal}{Journal of Symbolic Logic}} \bibinfo{volume}{56}, pp. \bibinfo{pages}{276--294}, \doi{10.2307/2274919}.

\bibitemdeclare{incollection}{Belnap77Computer}
\bibitem{Belnap77Computer}
\bibinfo{author}{N.~D. \surnamestart Belnap\surnameend} (\bibinfo{year}{1977}): \emph{\bibinfo{title}{How a computer should think}}.
\newblock In \bibinfo{editor}{G.~\surnamestart Ryle\surnameend}, editor: {\slshape \bibinfo{booktitle}{Contemporary Aspects of Philosophy}}, \bibinfo{publisher}{Oriel Press}, \bibinfo{address}{Stocksfield}, pp. \bibinfo{pages}{30--56}.

\bibitemdeclare{incollection}{Belnap77Useful}
\bibitem{Belnap77Useful}
\bibinfo{author}{N.~D. \surnamestart Belnap\surnameend} (\bibinfo{year}{1977}): \emph{\bibinfo{title}{A useful four-valued logic}}.
\newblock In \bibinfo{editor}{J.~M. \surnamestart Dunn\surnameend} \& \bibinfo{editor}{G.~\surnamestart Epstein\surnameend}, editors: {\slshape \bibinfo{booktitle}{Modern Uses of Multiple-Valued Logic}}, \bibinfo{publisher}{Kluwer Academic Publishers}, pp. \bibinfo{pages}{5--37}, \doi{10.1007/978-94-010-1161-7_2}.

\bibitemdeclare{article}{Dunn76}
\bibitem{Dunn76}
\bibinfo{author}{J.~M. \surnamestart Dunn\surnameend} (\bibinfo{year}{1976}): \emph{\bibinfo{title}{Intuitive semantics for first-degree entailments and `coupled trees'}}.
\newblock {\slshape \bibinfo{journal}{Philosophical Studies}} \bibinfo{volume}{29}(\bibinfo{number}{3}), pp. \bibinfo{pages}{149--168}, \doi{10.1007/BF00373152}.

\bibitemdeclare{article}{vanFraassen66}
\bibitem{vanFraassen66}
\bibinfo{author}{B.~C. \surnamestart van Fraassen\surnameend} (\bibinfo{year}{1966}): \emph{\bibinfo{title}{Singular Terms, Truth-Value Gaps, and Free Logic}}.
\newblock {\slshape \bibinfo{journal}{The Journal of Philosophy}} \bibinfo{volume}{63}(\bibinfo{number}{17}), pp. \bibinfo{pages}{481--495}, \doi{10.2307/2024549}.

\bibitemdeclare{incollection}{vanFraassen69}
\bibitem{vanFraassen69}
\bibinfo{author}{B.~C. \surnamestart van Fraassen\surnameend} (\bibinfo{year}{1969}): \emph{\bibinfo{title}{Presuppositions, Supervaluations, and Free Logic}}.
\newblock In \bibinfo{editor}{Karel \surnamestart Lambert\surnameend}, editor: {\slshape \bibinfo{booktitle}{The Logical Way of Doing Things}}, \bibinfo{publisher}{Yale University Press}, \bibinfo{address}{New Haven}, pp. \bibinfo{pages}{67--91}.

\bibitemdeclare{article}{Hintikka59}
\bibitem{Hintikka59}
\bibinfo{author}{J.~\surnamestart Hintikka\surnameend} (\bibinfo{year}{1959}): \emph{\bibinfo{title}{Existential presuppositions and existential commitments}}.
\newblock {\slshape \bibinfo{journal}{Journal of Philosophy}} \bibinfo{volume}{56}(\bibinfo{number}{3}), pp. \bibinfo{pages}{125--137}, \doi{10.2307/2021988}.

\bibitemdeclare{inproceedings}{IndrzejczakModalDescriptions18}
\bibitem{IndrzejczakModalDescriptions18}
\bibinfo{author}{A.~\surnamestart Indrzejczak\surnameend} (\bibinfo{year}{2018}): \emph{\bibinfo{title}{Cut-free Modal Theory of Definite Descriptions}}.
\newblock In \bibinfo{editor}{G.~\surnamestart Bezhanishvili\surnameend}, \bibinfo{editor}{G.~\surnamestart D'Agostino\surnameend}, \bibinfo{editor}{G.~\surnamestart Metcalfe\surnameend} \& \bibinfo{editor}{T.~\surnamestart Studer\surnameend}, editors: {\slshape \bibinfo{booktitle}{Advances in Modal Logic}}, \bibinfo{volume}{12}, \bibinfo{publisher}{College Publications}, \bibinfo{address}{London}, pp. \bibinfo{pages}{387--406}.

\bibitemdeclare{article}{IndrzejczakFreeDefDesc20}
\bibitem{IndrzejczakFreeDefDesc20}
\bibinfo{author}{A.~\surnamestart Indrzejczak\surnameend} (\bibinfo{year}{2020}): \emph{\bibinfo{title}{Free Definite Description Theory -- Sequent Calculi and Cut Elimination}}.
\newblock {\slshape \bibinfo{journal}{Logic and Logical Philosophy}} \bibinfo{volume}{29}(\bibinfo{number}{4}), pp. \bibinfo{pages}{505--539}, \doi{10.12775/LLP.2020.020}.

\bibitemdeclare{article}{IndrzejczakCutFree}
\bibitem{IndrzejczakCutFree}
\bibinfo{author}{A.~\surnamestart Indrzejczak\surnameend} (\bibinfo{year}{2021}): \emph{\bibinfo{title}{Free logics are cut free}}.
\newblock {\slshape \bibinfo{journal}{Studia Logica}} \bibinfo{volume}{109}(\bibinfo{number}{4}), pp. \bibinfo{pages}{859--886}, \doi{10.1007/s11225-020-09929-8}.

\bibitemdeclare{inproceedings}{IndrzejczakNegativeFreeLogics25}
\bibitem{IndrzejczakNegativeFreeLogics25}
\bibinfo{author}{A.~\surnamestart Indrzejczak\surnameend} (\bibinfo{year}{2025}): \emph{\bibinfo{title}{Cut Elimination for Negative Free Logics with Definite Descriptions}}.
\newblock In \bibinfo{editor}{C.~\surnamestart Barrett\surnameend} \& \bibinfo{editor}{U.~\surnamestart Waldmann\surnameend}, editors: {\slshape \bibinfo{booktitle}{Automated Deduction -- {CADE} 30}}, {\slshape \bibinfo{series}{Lecture Notes in Computer Science}} \bibinfo{volume}{15943}, \bibinfo{publisher}{Springer}, pp. \bibinfo{pages}{341--360}, \doi{10.1007/978-3-031-99984-0_19}.

\bibitemdeclare{inproceedings}{IndrzejczakPetrukhin24Bisequent}
\bibitem{IndrzejczakPetrukhin24Bisequent}
\bibinfo{author}{A.~\surnamestart Indrzejczak\surnameend} \& \bibinfo{author}{Y.~\surnamestart Petrukhin\surnameend} (\bibinfo{year}{2024}): \emph{\bibinfo{title}{Bisequent Calculi for Neutral Free Logic with Definite Descriptions}}.
\newblock In \bibinfo{editor}{C.~\surnamestart Benzm{\"u}ller\surnameend}, \bibinfo{editor}{J.~\surnamestart Otten\surnameend} \& \bibinfo{editor}{R.~\surnamestart Ramanayake\surnameend}, editors: {\slshape \bibinfo{booktitle}{Automated Reasoning in Quantified Non-Classical Logics}}, {\slshape \bibinfo{series}{CEUR Workshop Proceedings}} \bibinfo{volume}{3875}, pp. \bibinfo{pages}{48--61}.

\bibitemdeclare{book}{KamideWansing15}
\bibitem{KamideWansing15}
\bibinfo{author}{N.~\surnamestart Kamide\surnameend} \& \bibinfo{author}{H.~\surnamestart Wansing\surnameend} (\bibinfo{year}{2015}): \emph{\bibinfo{title}{Proof Theory of {N4}-related Paraconsistent Logics}}.
\newblock {\slshape \bibinfo{series}{Studies in Logic}}~\bibinfo{volume}{54}, \bibinfo{publisher}{College Publications}, \bibinfo{address}{Milton Keynes}.

\bibitemdeclare{article}{Kleene38}
\bibitem{Kleene38}
\bibinfo{author}{S.~\surnamestart Kleene\surnameend} (\bibinfo{year}{1938}): \emph{\bibinfo{title}{On notation for ordinal numbers}}.
\newblock {\slshape \bibinfo{journal}{Journal of Symbolic Logic}} \bibinfo{volume}{3}, pp. \bibinfo{pages}{150--155}, \doi{10.2307/2267778}.

\bibitemdeclare{article}{KurbisBinary19}
\bibitem{KurbisBinary19}
\bibinfo{author}{N.~\surnamestart K{\"u}rbis\surnameend} (\bibinfo{year}{2019}): \emph{\bibinfo{title}{A binary quantifier for definite descriptions in intuitionist negative free logic: natural deduction and normalisation}}.
\newblock {\slshape \bibinfo{journal}{Bulletin of the Section of Logic}} \bibinfo{volume}{48}(\bibinfo{number}{2}), pp. \bibinfo{pages}{81--97}, \doi{10.18778/0138-0680.48.2.01}.

\bibitemdeclare{article}{KurbisBinary20}
\bibitem{KurbisBinary20}
\bibinfo{author}{N.~\surnamestart K{\"u}rbis\surnameend} (\bibinfo{year}{2019}): \emph{\bibinfo{title}{Two treatments of definite descriptions in intuitionist negative free logic}}.
\newblock {\slshape \bibinfo{journal}{Bulletin of the Section of Logic}} \bibinfo{volume}{48}(\bibinfo{number}{4}), pp. \bibinfo{pages}{299--317}, \doi{10.18778/0138-0680.48.4.04}.

\bibitemdeclare{article}{KurbisBinary21}
\bibitem{KurbisBinary21}
\bibinfo{author}{N.~\surnamestart K{\"u}rbis\surnameend} (\bibinfo{year}{2021}): \emph{\bibinfo{title}{Definite descriptions in intuitionist positive free logic}}.
\newblock {\slshape \bibinfo{journal}{Logic and Logical Philosophy}} \bibinfo{volume}{30}(\bibinfo{number}{2}), pp. \bibinfo{pages}{327--358}, \doi{10.12775/LLP.2020.024}.

\bibitemdeclare{article}{KurbisIntuitionistic21}
\bibitem{KurbisIntuitionistic21}
\bibinfo{author}{N.~\surnamestart K{\"u}rbis\surnameend} (\bibinfo{year}{2021}): \emph{\bibinfo{title}{Normalisation and subformula property for a system of intuitionistic logic with general introduction and elimination rules}}.
\newblock {\slshape \bibinfo{journal}{Synthese}} \bibinfo{volume}{199}(\bibinfo{number}{5--6}), pp. \bibinfo{pages}{14223--14248}, \doi{10.1007/s11229-021-03418-8}.

\bibitemdeclare{inproceedings}{KurbisProofSemantics21}
\bibitem{KurbisProofSemantics21}
\bibinfo{author}{N.~\surnamestart K{\"u}rbis\surnameend} (\bibinfo{year}{2021}): \emph{\bibinfo{title}{Proof-Theory and Semantics for a Theory of Definite Descriptions}}.
\newblock In \bibinfo{editor}{A.~\surnamestart Das\surnameend} \& \bibinfo{editor}{S.~\surnamestart Negri\surnameend}, editors: {\slshape \bibinfo{booktitle}{Automated Reasoning with Analytic Tableaux and Related Methods}}, {\slshape \bibinfo{series}{Lecture Notes in Computer Science}} \bibinfo{volume}{12842}, \bibinfo{publisher}{Springer}, \bibinfo{address}{Cham}, pp. \bibinfo{pages}{95--111}, \doi{10.1007/978-3-030-86059-2_6}.

\bibitemdeclare{article}{KurbisBinary22}
\bibitem{KurbisBinary22}
\bibinfo{author}{N.~\surnamestart K{\"u}rbis\surnameend} (\bibinfo{year}{2022}): \emph{\bibinfo{title}{A Binary Quantifier for Definite Descriptions for Cut Free Free Logics}}.
\newblock {\slshape \bibinfo{journal}{Studia Logica}} \bibinfo{volume}{110}(\bibinfo{number}{1}), pp. \bibinfo{pages}{219--239}, \doi{10.1007/s11225-021-09958-x}.

\bibitemdeclare{article}{KurbisClassical22}
\bibitem{KurbisClassical22}
\bibinfo{author}{N.~\surnamestart K{\"u}rbis\surnameend} (\bibinfo{year}{2022}): \emph{\bibinfo{title}{Normalisation and subformula property for a system of classical logic with {Tarski's} rule}}.
\newblock {\slshape \bibinfo{journal}{Archive for Mathematical Logic}} \bibinfo{volume}{61}, pp. \bibinfo{pages}{105--129}, \doi{10.1007/s00153-021-00775-6}.

\bibitemdeclare{article}{KurbisNormNeg}
\bibitem{KurbisNormNeg}
\bibinfo{author}{N.~\surnamestart K{\"u}rbis\surnameend} (\bibinfo{year}{2025}): \emph{\bibinfo{title}{Normalisation for Negative Free Logics without and with Definite Descriptions}}.
\newblock {\slshape \bibinfo{journal}{The Review of Symbolic Logic}} \bibinfo{volume}{18}(\bibinfo{number}{1}), pp. \bibinfo{pages}{240--272}, \doi{10.1017/S1755020324000157}.

\bibitemdeclare{article}{KurbisNormPos}
\bibitem{KurbisNormPos}
\bibinfo{author}{N.~\surnamestart K{\"u}rbis\surnameend} (\bibinfo{year}{2026}): \emph{\bibinfo{title}{Normalisation for Positive Free Logics without and with Definite Descriptions}}.
\newblock {\slshape \bibinfo{journal}{The Review of Symbolic Logic}}, pp. \bibinfo{pages}{1--25}, \doi{10.1017/S1755020326101166}.
\newblock \bibinfo{note}{In press}.

\bibitemdeclare{article}{KurbisPetrukhin21}
\bibitem{KurbisPetrukhin21}
\bibinfo{author}{N.~\surnamestart K{\"u}rbis\surnameend} \& \bibinfo{author}{Y.~\surnamestart Petrukhin\surnameend} (\bibinfo{year}{2021}): \emph{\bibinfo{title}{Normalisation for some quite interesting many-valued logics}}.
\newblock {\slshape \bibinfo{journal}{Logic and Logical Philosophy}} \bibinfo{volume}{30}(\bibinfo{number}{3}), pp. \bibinfo{pages}{493--534}, \doi{10.12775/LLP.2021.009}.

\bibitemdeclare{article}{Lambert2}
\bibitem{Lambert2}
\bibinfo{author}{K.~\surnamestart Lambert\surnameend} (\bibinfo{year}{1961}): \emph{\bibinfo{title}{Notes on ``{E!}'': {II}}}.
\newblock {\slshape \bibinfo{journal}{Philosophical Studies}} \bibinfo{volume}{12}(\bibinfo{number}{1/2}), pp. \bibinfo{pages}{1--5}, \doi{10.1007/BF00714650}.

\bibitemdeclare{article}{Lambert3}
\bibitem{Lambert3}
\bibinfo{author}{K.~\surnamestart Lambert\surnameend} (\bibinfo{year}{1962}): \emph{\bibinfo{title}{Notes on ``{E!}'' {III}: A theory of descriptions}}.
\newblock {\slshape \bibinfo{journal}{Philosophical Studies}} \bibinfo{volume}{13}(\bibinfo{number}{4}), pp. \bibinfo{pages}{51--59}, \doi{10.1007/BF00462042}.

\bibitemdeclare{incollection}{Lambert4}
\bibitem{Lambert4}
\bibinfo{author}{K.~\surnamestart Lambert\surnameend} (\bibinfo{year}{2004}): \emph{\bibinfo{title}{Existential import, `{E!}' and `the'}}.
\newblock In: {\slshape \bibinfo{booktitle}{Free Logic. Selected Essays}}, \bibinfo{publisher}{Cambridge University Press}, \bibinfo{address}{Cambridge}, pp. \bibinfo{pages}{16--32}, \doi{10.1017/CBO9781139165068.003}.

\bibitemdeclare{article}{Lehmann94}
\bibitem{Lehmann94}
\bibinfo{author}{S.~\surnamestart Lehmann\surnameend} (\bibinfo{year}{1994}): \emph{\bibinfo{title}{Strict Fregean Free Logic}}.
\newblock {\slshape \bibinfo{journal}{Journal of Philosophical Logic}} \bibinfo{volume}{23}(\bibinfo{number}{3}), pp. \bibinfo{pages}{307--336}, \doi{10.1007/BF01048484}.

\bibitemdeclare{incollection}{Lehmann02}
\bibitem{Lehmann02}
\bibinfo{author}{S.~\surnamestart Lehmann\surnameend} (\bibinfo{year}{2002}): \emph{\bibinfo{title}{More Free Logic}}.
\newblock In: {\slshape \bibinfo{booktitle}{Handbook of Philosophical Logic}}, \bibinfo{volume}{5}, \bibinfo{publisher}{Springer}, pp. \bibinfo{pages}{197--259}, \doi{10.1007/978-94-017-0458-8_4}.

\bibitemdeclare{article}{Milne10}
\bibitem{Milne10}
\bibinfo{author}{P.~\surnamestart Milne\surnameend} (\bibinfo{year}{2010}): \emph{\bibinfo{title}{Subformula and separation properties in natural deduction via small {Kripke} models}}.
\newblock {\slshape \bibinfo{journal}{Review of Symbolic Logic}} \bibinfo{volume}{3}(\bibinfo{number}{2}), pp. \bibinfo{pages}{175--227}, \doi{10.1017/S175502030999030X}.

\bibitemdeclare{book}{NegriPlato01}
\bibitem{NegriPlato01}
\bibinfo{author}{S.~\surnamestart Negri\surnameend} \& \bibinfo{author}{J.~\surnamestart von Plato\surnameend} (\bibinfo{year}{2001}): \emph{\bibinfo{title}{Structural Proof-Theory}}.
\newblock \bibinfo{publisher}{Cambridge University Press}, \bibinfo{address}{Cambridge}, \doi{10.1017/CBO9780511527340}.

\bibitemdeclare{article}{PavlovicGratzl23}
\bibitem{PavlovicGratzl23}
\bibinfo{author}{E.~\surnamestart Pavlovi{\'c}\surnameend} \& \bibinfo{author}{N.~\surnamestart Gratzl\surnameend} (\bibinfo{year}{2023}): \emph{\bibinfo{title}{Neutral Free Logic: Motivation, Proof Theory and Models}}.
\newblock {\slshape \bibinfo{journal}{Journal of Philosophical Logic}} \bibinfo{volume}{52}, pp. \bibinfo{pages}{519--554}, \doi{10.1007/s10992-022-09679-z}.

\bibitemdeclare{inproceedings}{PetrukhinNelsonianDescriptions}
\bibitem{PetrukhinNelsonianDescriptions}
\bibinfo{author}{Y.~\surnamestart Petrukhin\surnameend} (\bibinfo{year}{2024}): \emph{\bibinfo{title}{A binary quantifier for definite descriptions in {Nelsonian} free logic}}.
\newblock In \bibinfo{editor}{A.~\surnamestart Indrzejczak\surnameend} \& \bibinfo{editor}{M.~\surnamestart Zawidzki\surnameend}, editors: {\slshape \bibinfo{booktitle}{Proceedings Eleventh International Conference on Non-Classical Logics. Theory and Applications}}, {\slshape \bibinfo{series}{Electronic Proceedings in Theoretical Computer Science}} \bibinfo{volume}{415}, pp. \bibinfo{pages}{5--15}, \doi{10.4204/EPTCS.415.5}.

\bibitemdeclare{inproceedings}{PetrukhinJelia}
\bibitem{PetrukhinJelia}
\bibinfo{author}{Y.~\surnamestart Petrukhin\surnameend} (\bibinfo{year}{2026}): \emph{\bibinfo{title}{On a Second-Order Version of Russellian Theory of Definite Descriptions}}.
\newblock In \bibinfo{editor}{G.~\surnamestart Casini\surnameend}, \bibinfo{editor}{B.~\surnamestart Dundua\surnameend} \& \bibinfo{editor}{T.~\surnamestart Kutsia\surnameend}, editors: {\slshape \bibinfo{booktitle}{Logics in Artificial Intelligence. JELIA 2025}}, {\slshape \bibinfo{series}{Lecture Notes in Computer Science}} \bibinfo{volume}{16094}, \bibinfo{publisher}{Springer}, \bibinfo{address}{Cham}, pp. \bibinfo{pages}{313--326}, \doi{10.1007/978-3-032-04590-4_21}.

\bibitemdeclare{inproceedings}{Petrukhin2026PRL}
\bibitem{Petrukhin2026PRL}
\bibinfo{author}{Y.~\surnamestart Petrukhin\surnameend} (\bibinfo{year}{2026}): \emph{\bibinfo{title}{A Paradefinite Version of Russellian Theory of Definite Descriptions}}.
\newblock In: {\slshape \bibinfo{booktitle}{Proceedings of the 6th International Workshop on Automated Reasoning in Quantified Non-Classical Logics (ARQNL 2026)}}, \bibinfo{series}{CEUR Workshop Proceedings}.
\newblock \urlprefix\url{https://iltp.de/ARQNL-2026/download/proceedings_preli/4_ARQNL_2026_paper_2.pdf}.

\bibitemdeclare{inproceedings}{Petrukhin2026SO}
\bibitem{Petrukhin2026SO}
\bibinfo{author}{Y.~\surnamestart Petrukhin\surnameend} (\bibinfo{year}{2026}): \emph{\bibinfo{title}{Second-order Positive Free Logic with Second-order Definite Descriptions}}.
\newblock In: {\slshape \bibinfo{booktitle}{Proceedings of the 6th International Workshop on Automated Reasoning in Quantified Non-Classical Logics (ARQNL 2026)}}, \bibinfo{series}{CEUR Workshop Proceedings}.
\newblock \urlprefix\url{https://iltp.de/ARQNL-2026/download/proceedings_preli/6_ARQNL_2026_paper_1.pdf}.

\bibitemdeclare{article}{PetrukhinShanginK3}
\bibitem{PetrukhinShanginK3}
\bibinfo{author}{Y.~\surnamestart Petrukhin\surnameend} \& \bibinfo{author}{V.~\surnamestart Shangin\surnameend} (\bibinfo{year}{2019}): \emph{\bibinfo{title}{Automated Proof-searching for Strong Kleene Logic and its Binary Extensions via Correspondence Analysis}}.
\newblock {\slshape \bibinfo{journal}{Logic and Logical Philosophy}} \bibinfo{volume}{28}(\bibinfo{number}{2}), pp. \bibinfo{pages}{223--257}, \doi{10.12775/LLP.2018.009}.

\bibitemdeclare{article}{Petrukhin19}
\bibitem{Petrukhin19}
\bibinfo{author}{Y.~I. \surnamestart Petrukhin\surnameend} (\bibinfo{year}{2019}): \emph{\bibinfo{title}{Deduction normalization theorem for {Sette's} logic and its modifications}}.
\newblock {\slshape \bibinfo{journal}{Moscow University Mathematics Bulletin}} \bibinfo{volume}{74}(\bibinfo{number}{1}), pp. \bibinfo{pages}{25--31}, \doi{10.3103/S0027132219010054}.

\bibitemdeclare{incollection}{Popov89}
\bibitem{Popov89}
\bibinfo{author}{V.~M. \surnamestart Popov\surnameend} (\bibinfo{year}{1989}): \emph{\bibinfo{title}{Sequent formulations of paraconsistent logical systems}}.
\newblock In: {\slshape \bibinfo{booktitle}{Semantic and Syntactic Investigations of Non-Extensional Logics}}, \bibinfo{publisher}{Nauka}, pp. \bibinfo{pages}{285--289}.
\newblock \bibinfo{note}{In Russian}.

\bibitemdeclare{inproceedings}{Popov09}
\bibitem{Popov09}
\bibinfo{author}{V.~M. \surnamestart Popov\surnameend} (\bibinfo{year}{2009}): \emph{\bibinfo{title}{Between the logic par and the set of all formulas}}.
\newblock In: {\slshape \bibinfo{booktitle}{The Proceeding of the Sixth Smirnov Readings in Logic}}, \bibinfo{publisher}{Contemporary Notebooks}, pp. \bibinfo{pages}{93--95}.
\newblock \bibinfo{note}{In Russian}.

\bibitemdeclare{book}{Prawitz}
\bibitem{Prawitz}
\bibinfo{author}{D.~\surnamestart Prawitz\surnameend} (\bibinfo{year}{1965}): \emph{\bibinfo{title}{Natural Deduction}}.
\newblock \bibinfo{publisher}{Almqvist and Wiksell}, \bibinfo{address}{Stockholm, G{\"o}teborg, Uppsala}.

\bibitemdeclare{incollection}{Priest02}
\bibitem{Priest02}
\bibinfo{author}{G.~\surnamestart Priest\surnameend} (\bibinfo{year}{2002}): \emph{\bibinfo{title}{Paraconsistent logic}}.
\newblock In: {\slshape \bibinfo{booktitle}{Handbook of Philosophical Logic}}, \bibinfo{volume}{6}, \bibinfo{publisher}{Kluwer Academic Publishers}, pp. \bibinfo{pages}{287--393}, \doi{10.1007/978-94-017-0460-1_4}.

\bibitemdeclare{article}{SlupeckiPComp}
\bibitem{SlupeckiPComp}
\bibinfo{author}{E.~\surnamestart S{\l}upecki\surnameend}, \bibinfo{author}{J.~\surnamestart Bryll\surnameend} \& \bibinfo{author}{T.~\surnamestart Prucnal\surnameend} (\bibinfo{year}{1967}): \emph{\bibinfo{title}{Some remarks on the three-valued logic of J. {\L}ukasiewicz}}.
\newblock {\slshape \bibinfo{journal}{Studia Logica}} \bibinfo{volume}{21}, pp. \bibinfo{pages}{45--70}, \doi{10.1007/BF02123418}.

\bibitemdeclare{book}{Tennant78}
\bibitem{Tennant78}
\bibinfo{author}{N.~\surnamestart Tennant\surnameend} (\bibinfo{year}{1978}): \emph{\bibinfo{title}{Natural Logic}}.
\newblock \bibinfo{publisher}{Edinburgh University Press}, \bibinfo{address}{Edinburgh}.

\bibitemdeclare{book}{TroelstraSchwichtenberg}
\bibitem{TroelstraSchwichtenberg}
\bibinfo{author}{A.~S. \surnamestart Troelstra\surnameend} \& \bibinfo{author}{H.~\surnamestart Schwichtenberg\surnameend} (\bibinfo{year}{2000}): \emph{\bibinfo{title}{Basic Proof Theory}}, \bibinfo{edition}{second} edition.
\newblock {\slshape \bibinfo{series}{Cambridge Tracts in Theoretical Computer Science}}~\bibinfo{volume}{43}, \bibinfo{publisher}{Cambridge University Press}, \bibinfo{address}{Cambridge}, \doi{10.1017/CBO9781139168717}.

\bibitemdeclare{incollection}{Woodruff70}
\bibitem{Woodruff70}
\bibinfo{author}{P.~W. \surnamestart Woodruff\surnameend} (\bibinfo{year}{1970}): \emph{\bibinfo{title}{Logic and Truth Value Gaps}}.
\newblock In \bibinfo{editor}{Karel \surnamestart Lambert\surnameend}, editor: {\slshape \bibinfo{booktitle}{Philosophical Problems in Logic: Some Recent Developments}}, \bibinfo{publisher}{D. Reidel Publishing Company}, \bibinfo{address}{Dordrecht}, pp. \bibinfo{pages}{121--142}, \doi{10.1007/978-94-010-3272-8_6}.

\bibitemdeclare{article}{Woodruff84}
\bibitem{Woodruff84}
\bibinfo{author}{P.~W. \surnamestart Woodruff\surnameend} (\bibinfo{year}{1984}): \emph{\bibinfo{title}{On Supervaluations in Free Logic}}.
\newblock {\slshape \bibinfo{journal}{The Journal of Symbolic Logic}} \bibinfo{volume}{49}, pp. \bibinfo{pages}{943--950}, \doi{10.2307/2274148}.

\bibitemdeclare{article}{Zimmermann}
\bibitem{Zimmermann}
\bibinfo{author}{E.~\surnamestart Zimmermann\surnameend} (\bibinfo{year}{2002}): \emph{\bibinfo{title}{{Peirce's} rule in natural deduction}}.
\newblock {\slshape \bibinfo{journal}{Theoretical Computer Science}} \bibinfo{volume}{275}, pp. \bibinfo{pages}{561--574}, \doi{10.1016/S0304-3975(01)00296-1}.

\end{thebibliography}
\end{document}